\documentclass[preprint]{elsarticle}

\usepackage{lineno,hyperref}
\modulolinenumbers[5]

\usepackage{amsmath,amssymb,amsthm,mathtools,bm,tikz-cd,xfrac} 

\usepackage{varioref} 

\usepackage{comment}

\usepackage{pgfplots,tikz}
\usetikzlibrary{calc}
\usetikzlibrary{arrows.meta,positioning}
\usepackage{xfrac}
\usepackage{todonotes}

\newcommand{\R}{\mathbb{R}}

\newcommand{\calA}{\mathcal{A}}

\newcommand{\calC}{\mathcal{C}}
\newcommand{\calD}{\mathcal{D}}

\newcommand{\calL}{\mathcal{L}}

\newcommand{\calP}{\mathcal{P}}

\newcommand{\calS}{\mathcal{S}}
\newcommand{\calT}{\mathcal{T}}

\newcommand{\calX}{\mathcal{X}}
\newcommand{\calY}{\mathcal{Y}}

\newcommand{\Prob}{\mathbb{P}}
\newcommand{\sfG}{\mathsf{G}}

\DeclareMathOperator*{\argmax}{arg\,max}

\renewcommand{\phi}{\varphi}
\renewcommand{\epsilon}{\varepsilon}

\newcommand{\bx}{\bm{x}}
\newcommand{\bX}{\bm{X}}

\def\clap#1{\hbox to 0pt{\hss#1\hss}}

\newcommand{\braket}[1]{\langle #1 \rangle}
\renewcommand{\phi}{\varphi}

\theoremstyle{definition} 
\newtheorem{definition}{Definition}

\theoremstyle{plain} 
\newtheorem{proposition}{Proposition}
\newtheorem{lemma}{Lemma}

\newtheorem{example}{Example}
\newtheorem{remark}{Remark}

\theoremstyle{remark} 
\newtheorem{convention}{Convention}

\begin{document}

\begin{frontmatter}

\title{Values of Games for Information Decomposition}

\author{Tomáš Kroupa}

\ead{tomas.kroupa@fel.cvut.cz}

\author{Sara Vannucci}

\ead{vanucsar@fel.cvut.cz}

\author{Tomáš Votroubek}

\ead{votroto1@fel.cvut.cz}

\address{Artificial Intelligence Center, Department of Computer Science,\\ Faculty of Electrical Engineering, Czech Technical University in~Prague,\\Karlovo náměstí 13, 12O OO Praha 2,  Czech Republic}

\begin{abstract}
    The information decomposition problem requires an additive decomposition of the mutual information between the input and target variables into nonnegative terms. The recently introduced solution to this problem, Information Attribution, involves the Shapley-style value measuring the influence of predictors in the coalitional game associated with the joint probability distribution of the input random vector and the target variable. Motivated by the original problem, we consider a general setting of coalitional games where the players form a boolean algebra, and the coalitions are the corresponding down-sets. This enables us to study in detail various single-valued solution concepts, called values. Namely, we focus on the classes of values that can represent very general alternatives to the solution of the information decomposition problem, such as random-order values or sharing values. We extend the axiomatic characterization of some classes of values that were known only for the standard coalitional games.
\end{abstract}


\end{frontmatter}


\section{Introduction}
Measures of information content are widely used in statistics, artificial intelligence, and machine learning. The Shannon entropy, Kullback-Leibler divergence, and other information-theoretic quantities are highly instrumental in learning or fitting parameters to data. Several extensions of such information measures have been introduced to tackle the problems arising in alternative uncertainty calculi as well. The recent paper \cite{Bouchon} discusses axiomatic aspects of generalized information measures. As for the applications to machine learning, a new feature selection method based on an alternative to the joint mutual information is designed, and its performance is evaluated in \cite{Salem}. Mutual information is an essential tool in modern methods for explaining the predictive power of features in complex black-box classifiers such as deep neural networks. A case in point is the method SAGE (Shapley Additive Global importancE) \cite{Covert} in which the features used in the classifier represent the players in a particular game scenario. The games induced by machine learning problems belong to the area of cooperative game theory \cite{MaschlerSolanZamir13,PelegSudholter07}. There are many variations in coalitional games used for such problems. Some researchers employ the ``local'' approach where the influence of features is evaluated with respect to a fixed configuration of the target variable and the features \cite{Strumbelj14}. Moreover, assessing the contribution of features by a value-like concept from game theory seems so natural that it has been discovered independently without a reference to any specific game; see \cite{somol2011fast}.

 The main principles underlying SAGE are to construct a coalitional game assessing the coalitions of features for the given classifier and, subsequently, to compute the Shapley value of individual features. This computation is in fact the approximation based on a sampling algorithm since the underlying game typically has a huge number of players. The Shapley value is thus instrumental in evaluating the predictive power of features. Moreover, the well-known axioms of Shapley value \cite{Shapley53} provide a natural and domain-independent justification for such a distribution of the ``total'' predictive capacity among the individual features. 

The recent paper \cite{AyPolaniVirgo20} combines information- and game-theoretic tools to solve the following information decomposition problem:
\begin{quote}
    How to decompose additively the mutual information between input random variables and the target variable into the sum of nonnegative terms quantifying the information contribution of each set of input variables?
\end{quote}
The information decomposition problem asks for measuring the influence of sets of input variables and not only that of individual variables (features). The need to capture such complex and potentially high-dimensional interactions arises already from very simple probabilistic models where two-dimensional interactions can't describe the behavior of the system faithfully. The folklore example is the \textsc{XOR} model, in which two independent random bits $X_1$ and $X_2$ interact three-way with $Y=X_1\oplus X_2$, but neither $X_1$ nor $X_2$ alone interacts pairwise with $Y$. The solution to the information decomposition problem proposed in \cite{AyPolaniVirgo20} is called Information Attribution. It is worth emphasizing that the passage from measuring contributions of individual variables to quantifying the contributions of sets of such variables is the key difference of Information Attribution from the methods such as above mentioned SAGE \cite{Covert}. On the one hand, this characteristic makes Information Attribution much more expressive. On the other, it is computationally challenging to deal with already moderate-size models using this technique due to the exponentially increasing number of players.

In the rest of this section, first we will state our assumptions (Section \ref{sec:assum}) and then define the information decomposition problem formally (Section \ref{sec:IDoutline}). Since the special game-theoretic framework of Information Attribution was the main motivation for the problems studied in this paper, we briefly sketch the basic principles of Information Attribution in Section \ref{sec:AI}; the full-blown description can be found in \ref{a:IA}. The interested reader is invited to consult the original paper \cite{AyPolaniVirgo20} for more details. Our methodology is explained in Section \ref{sec:our}.

\subsection{The assumptions}\label{sec:assum}
 All random variables are discrete in this paper. We consider $n$ \emph{input} random variables $X_1,\dots,X_n$ each of which has a finite state space $\calX_i$ and a \emph{target} random variable\footnote{Equivalently, the random variable $Y$ might be called the \emph{response variable} or \emph{label}.} $Y$ with a finite state space $\calY$. Let $\calX=\calX_1\times\dots\times\calX_n$. The random vector $(\bX,Y)=(X_1,\dots,X_n,Y)$ on $\calX\times\calY$ is distributed according to the known joint probability distribution $p=p_{\bX Y}$, which is  called the \emph{true distribution} in this context. We will also deal with $n$-dimensional marginal distribution $p_{\bX}$ of $\bX$ and one-dimensional marginal distribution $p_Y$ of $Y$. 

The interaction between $\bX$ and $Y$ can be captured by various information-theoretic quantities; see \cite{Ay18,CoverThomas06}. For example, we can average out the values $\bx\in\calX$ of input random vector $\bX$ which provide some information about the values $y\in\calY$ of target variable $Y$. This leads to the concept of mutual information betweeen $\bX$ and $Y$. Specifically, the \emph{mutual information} $I(\bX;Y)$ is the expected amount of information about $Y$ provided by $\bX$, that is,
\begin{equation}\label{def:JMI}
    I(\bX;Y) = \sum_{(\bx,y)\in\calX\times \calY} p(\bx,y)\cdot \log_2\frac{p(\bx,y)}{p_{\bX}(\bx)\cdot p_Y(y)}.
\end{equation}
These are the basic properties of $I(\bX;Y)$.
\begin{itemize}
    \item $I(\bX;Y)\geq 0$.
    \item $I(\bX;Y)= 0$ if, and only if, $\bX$ and $Y$ are independent.
\end{itemize}
Note that independence of $\bX$ and $Y$ means precisely that $p(\bx,y)=  p_{\bX}(\bx)\cdot p_Y(y)$ for all $(\bx,y)\in\calX\times \calY$.

\subsection{Information decomposition problem}\label{sec:IDoutline}
Our presentation is based on \cite[Section 1]{AyPolaniVirgo20}. Define $V=\{1,\dots,n\}$ and let $A\subseteq V$. A \emph{predictor}\footnote{Note that this concept of predictor is different from the one used in machine learning and statistics, where ``predictor'' is the prediction function for observed data.} is any set of input random variables $\{X_i\mid i\in A\}$. We will often call the set $A$ a predictor, too. Thus the set of all predictors, which is precisely the set of all subsets of $V$, is denoted by $\calP_V$. The information decomposition problem informally introduced above can be stated precisely as follows.
\begin{quote}
    For an arbitrary joint probability distribution $p_{\bX Y}$ of $n$ random inputs $\bX=(X_1,\dots,X_n)$ and target variable $Y$, find real numbers $\psi_A\geq 0$ such that \begin{equation}\label{def:IDinf} I(\bX;Y)=\sum_{A\in\calP_V}\psi_A.\end{equation}
\end{quote}
Additional properties of the decomposition are usually required to obtain a sensible solution; we refer the interested reader to \cite[Section 1]{AyPolaniVirgo20} for the thorough discussion of the motivation and history of the information decomposition problem. In particular, the nonnegative numbers $\psi_A$ are interpreted as ``the proportion of the total mutual information attributed to predictor $A$, beyond what is already provided by its subsets''. It is always assumed that the information attributed to the empty predictor $\emptyset$ is zero, $\psi_{\emptyset}=0$.

\subsection{Information Attribution, briefly}\label{sec:AI}
 The solution to the information decomposition problem based on cooperative game theory was presented in \cite{AyPolaniVirgo20}. The proposed methodology is called \emph{Information Attribution} and these are its main ideas.
 \begin{enumerate}
    \item The class of all predictors $\calP_V$ uniquely generates a certain composite system $\calD(\calP_V,\subseteq)$. Specificaly, $\calD(\calP_V,\subseteq)$ is the lattice of down-sets of predictors with respect to the set-theoretic inclusion $\subseteq$ in $\calP_V$. In other words $\calD(\calP_V,\subseteq)$ contains exactly those sets of predictors which are closed with respect to the inclusion. From the game-theoretic viewpoint, the predictors (elements of $\calP_V$) are interpreted as \emph{players} and the down-sets of predictors (elements of $\calD(\calP_V,\subseteq)$) become \emph{coalitions}. It is worth mentioning that the choice of down-sets as coalitions is very natural here since it is a consequence of the information-theoretic interpretation of the problem.
    \item A \emph{coalitional game} $v_p\colon \calD(\calP_V,\subseteq)\to \R$  is induced by the true distribution $p$ of random vector $(\bX,Y)$ in a particular way. In this context, the worth $v_p(\calS)$ of each coalition $\calS\in\calD(\calP_V,\subseteq)$ is defined to be the information distance of $p$ from a certain baseline distribution depending on $p$ and $\calS$. The number $v_p(\calS)$ can be intepreted as a measure of complexity of the true distribution $p$ as in \cite[Chapter 6.1]{Ay18}.
    \item The coalitional game $v_p$ is used to find the value $\psi_A$ of each predictor $A\in \calP_V$. This is based on the concept of hierachical value of Faigle and Kern \cite{FaigleKern92}. The hierachical values was developed as the extension of Shapley value \cite{Shapley53} to the coalitional games with a partially ordered player set. The partial order captures the hierarchy among the players. The values $\psi_A$ are then declared to be the solution to the information decomposition problem \eqref{def:IDinf}.
 \end{enumerate}

 \subsection{Outline and methodology}\label{sec:our}
In this paper we study the general game-theoretic model originating from  Information Attribution. The first step is to replace the powerset of predictors with any finite boolean algebra of players. This enables us to present the original framework in a lighter notation and to connect our formulation to the known results about coalitional games with restricted coalition formation. In Section \ref{sec:coal} we discuss basic properties of coalitional games where players form the boolean algebra and the coalitions are restricted to the corresponding lattice of down-sets. This also involves the introduction of Harsanyi dividends (M\"obius transform), one of the key tools instrumental in the development of value theory for such games. Our approach to the study of all coalitional games instead of just the ``information games'' induced by true distributions $p_{\bX Y}$ might be seen as too general at the first sight. However, we claim that it is not only beneficial to formulate the original problem in the most general way, but it is also necessary. This is for the following reasons.
\begin{itemize}
    \item The solution of information decomposition problem \eqref{def:IDinf} employs the hierachical value \cite{FaigleKern92}, which is based on the uniform distribution over admissible permutations of players. However, it was also pointed out in \cite{AyPolaniVirgo20} that such a distribution might not be the only sensible choice.
    \item The coalitional games $v_p$ arising from Information Attribution do not have any extra properties except monotonicity.
    \item On top of that, the class of all such games does not form any subset amenable to the methods of cooperative game theory. Specifically, it can be shown the set of coalitional games $v_p$ obtained by varying all probability distributions $p$ of the random vector $(\bX,Y)$ is compact but not convex. Therefore, it is not a linear space.
\end{itemize}

Our main results are in Section \ref{sec:val}. The key observation is that the hierachical value used in Information Attribution is not the only value concept that can be applied to the solution of information decomposition problem. By a \emph{value} we mean an operator mapping a family of coalitional games to the set of possible vector allocations for individual players. Therefore, we focus our study on the class of efficient and positive values for coalitional games where players form a boolean algebra, since the two properties of values are the minimal conditions for the solution of information decomposition problem \eqref{def:IDinf}. It will become clear later in Section \ref{sec:val} that Efficiency means the existence of additive decomposition \eqref{def:IDinf} and Positivity is the nonnegativity of $\psi_A$ in \eqref{def:IDinf}. This will lead naturally to the study of \emph{random-order values} and positive \emph{sharing values}. The former class is based on the average evaluation of player's contribution across the possible coalitions, whereas the latter involves averaging over Harsanyi dividends of the coalitions to which the player belongs.  We provide the characterization of those classes and discuss some examples appearing in the literature.

Section \ref{sec:concl} summarizes our main findings and mentions several items for further research. Appendix contains the description of coalitional games over $2$- and $4$-player boolean algebras (\ref{a:4} and \ref{a:8}, respectively) and the detailed specification of Information Attribution (\ref{a:IA}).

\section{Coalitional Games With the Boolean Algebra of Players}\label{sec:coal}
The \emph{standard model} of a coalitional game is based on the assumption that the player set is trivially ordered --- there is no relation between any pair of different players. This assumption means that a feasible coalition can be any subset of the player set; see \cite{MaschlerSolanZamir13,PelegSudholter07} for the thorough exposition of the standard model of coalitional games. In this section, we introduce coalitional games where the player set is a boolean algebra\footnote{A word of caution is in order here. The boolean algebra plays the role of player set in our setting, whereas the algebraically identical concept of powerset represents the set of all coalitions in the standard model of coalitional games.}, and we will further discuss the concepts from coalitional games related to the Shapley value over partially ordered player sets; see \cite{FaigleKern92,Algaba17,Algaba19}. As for the elements of lattice and order theory used here, we refer the reader to \cite{Stanley2012,Grabisch16} for all the unexplained notions.
\subsection{Players and coalitions}
Let $P$ be a nonempty finite set. An element of $P$ is called a \emph{player} and it is usually denoted by small letters such as $a,b,\dots,i,j,\dots$ Throughout the paper we always assume that $P$ is partially ordered by $\preceq$ such that the pair $(P,\preceq)$ is a boolean algebra. The order $\preceq$ expresses precedence constrains of hierarchy among the players in $P$ \cite{FaigleKern92,Algaba19}. The boolean algebra $P$ has $|P|=2^n$ elements, where $n$ is the number of atoms of $P$. The join and meet in $(P,\preceq)$ are denoted by $\vee$ and $\wedge$, respectively. By the \emph{rank} of an element $i\in P$ we mean the number $\rho(i)$ of uniquely determined distinct atoms $a_1,\dots,a_{\rho(i)}$ such that $i=a_1\vee \dots \vee a_{\rho(i)}$. In particular, the top element $\top$ (the highest player in the hierarchy) has rank $n$, and we also say that the boolean algebra $(P,\preceq)$ has rank $n$. The bottom element $\bot$ (the lowest player in the hierarchy) has rank $0$. It is well-known that
\begin{itemize}
    \item the number of atoms $n$ uniquely determines any finite boolean algebra up to an isomorphism and
    \item each boolean algebra $(P,\preceq)$ of rank $n$ is isomorphic to the powerset of an $n$-element set.
\end{itemize}

The usual concept of permutation of players needs to be generalized so that the resulting permutation is compatible with the order of players in $P$ given by~$\preceq$. This leads to the following concept. A mapping $f\colon P\to \{1,\dots,2^n\}$ is a \emph{linear extension}\footnote{Equivalently, linear extensions are also called admissible permutations.} of $(P,\preceq)$ if $f$ is bijective and order-preserving. Specifically, the latter property says that $i\preceq j$ implies $f(i)\leq f(j)$, for all $i,j\in P$, where $\leq$ denotes the usual total order between natural numbers. We can think of $f$ as the ranking of a player set which respects the structure of superiors. Let $\calL(P)$ be the set of all linear extensions of $(P,\preceq)$. Enumerating linear extensions of $\calL(P)$ is a hard problem -- see \cite{Brightwell03} for non-trivial bounds on $|\calL(P)|$. For example, it is known that $|\calL(P)| \ge 1.5\times 10^6$ for the boolean algebra $P$ of rank $n=4$. We will need the following criterion of equality of $f_1,f_2\in \calL(P)$. Namely $f_1=f_2$ if, and only if, the condition 
\begin{equation}\label{equality_crit}
    \text{If $f_1(i)> f_1(j)$, then $f_2(i)> f_2(j)$, \quad for all $i,j\in P$}
\end{equation}
holds true. The non-trivial implication in the above equivalence is based on this observation. For any $i \in P$ and any $f\in\calL(P)$, the number $f(i)$ is uniquely determined by the cardinality of the set $\{j \in P \mid f(i) > f(j)\}$. Since such sets corresponding to $f_1$ and $f_2$ coincide by the assumption \eqref{equality_crit}, we get $f_1=f_2$.

Let $S\subseteq P$ be a subset of players. An element $i\in S$ is a \emph{maximal element} of $S$ if there is no element $j\in S\setminus i$ with $i\preceq j$. Let $S^*\subseteq S$ be the set of all maximal elements of $S$. Since $P$ is finite, the set $S^*$ is nonempty whenever $S$ is. A~subset $S\subseteq P$ is
\begin{itemize}
    \item a \emph{down-set} if $i\in S$, $j\in P$, and $j\preceq i$ implies $j\in S$,
    \item an \emph{antichain} if $i\not\preceq j$ and $j\not\preceq j$ for all $i,j\in S$.
\end{itemize}
For any $S\subseteq P$, define $\braket{S} =\{j\in P \mid j\preceq i \text{ for some $i\in S$}\}$. Then $\braket{S}$ is a down-set. If $S$~is an antichain, then $S=\langle S\rangle^*$. Conversely, the set $S^*$ of maximal elements of a down-set $S$ is necessarily an antichain and $S=\langle S^*\rangle$. This yields a one-to-one correspondence between antichains and down-sets in $(P,\preceq)$. We will frequently use this observation: If $S\subseteq P$ is a down-set and $i\in S^*$, then $S\setminus \{i\}$ is a down-set, too. 

In order for a \emph{coalition} $S\subseteq P$ to be feasible with respect to the precedence constraints given by $\preceq$, we consider only those $S$ containing all the subordinates of each superior player in $S$. This means exactly that a coalition $S$ must be a down-set in $(P,\preceq)$. In the paper we will use the terms ``down-set'' and ``feasible'' interchangeably. The set of all down-sets $\calD(P,\preceq)$ is a finite distributive lattice in which the join and meet coincide with the set-theoretic operations $\cup$ and $\cap$, respectively. The partial order of $\calD(P,\preceq)$ is the inclusion $\subseteq$ between down-sets. Interestingly enough, $\calD(P,\preceq)$ is very special among the distributive lattices, since it is precisely the free distributive lattice generated by the $n$ atoms of $(P,\preceq)$. The role of this lattice for information theory is further discussed in \cite{AyPolaniVirgo20}.

\begin{convention}
    We will frequently omit the curly braces. For example, suppose $b,e\in P$. Expressions such as $\langle b,e\rangle$ and $P\setminus e$ are understood as $\langle \{b,e\}\rangle$ and $P\setminus \{e\}$, respectively. 
\end{convention}

See Figures \ref{FigBA4} and \ref{FigBA8} for the examples of boolean algebra $P$ with ranks $n=2,3$ with the corresponding $\calD(P,\preceq)$. The rapidly growing cardinalities of $\calD(P,\preceq)$ are  shown in Table \ref{tabDed}. The resulting numbers are the Dedekind numbers. By the definition, each Dedekind number is the number of antichains in the powerset of an $n$-element set.

\begin{table}[h]
    \begin{center}
    \begin{tabular}{|c|| c | c | c | c | c|}
        \hline			
     $n$ & $2$ & $3$ & $4$  & $5$ & $6$\\
    $|(P,\preceq)|$ & $4$ & $8$ & $16$ & $32$ & $64$\\
     $|\calD(P,\preceq)|$ & $6$ & $20$ & $168$ & $7\,581$ & $7\,828\,354$\\
        \hline  
      \end{tabular} 
      \caption{The cardinalities of player set and coalition set for different ranks $n$}
      \label{tabDed}
    \end{center}
\end{table}

\subsection{Coalitional games}\label{sec:coalprop}
A coalitional game assigns to each feasible coalition the amount of utility as the result of cooperation among the members of the coalition. We will write simply $\calD$ in place of $\calD(P,\preceq)$. A \emph{(coalitional) game} on~$\calD$ is a function $v\colon \calD\to \R$ such that $v(\emptyset)=0$. The function $v$ associates with each feasible coalition $S\in\calD$ its worth $v(S)\in\R$. We adopt the standard concepts of coalitional game theory in the setting of games with restricted cooperation; see \cite{Grabisch16} for a survey and \cite{GrabischKroupa19} for the setting of partially ordered player set in particular. A game $v$ is called
\begin{itemize}
    \item \emph{monotone} if $S\subseteq T$ implies $v(S)\leq v(T)$, 
    \item \emph{nonnegative} when $v(S)\geq 0$,
    \item \emph{supermodular} if $v(S)+v(T) \leq v(S\cup T) + v(S\cap T)$,
    \item \emph{submodular} if $v(S)+v(T) \geq v(S\cup T) + v(S\cap T)$,
\end{itemize}
for all $S,T\in\calD$. A player $i\in P$ is \emph{null}\footnote{Note that a null player in this sense is called a ``dummy player'' in \cite[Example 3]{FaigleKern92}. However, dummy players are usually defined by a weaker condition in the game-theoretic literature.} in game $v$ if $v(S)=v(S\cup i)$, for all $S\in \calD$ such that $S\cup i\in\calD$. A \emph{carrier} for a game $v$ is a coalition $U\in\calD$ satisfying $v(S)=v(S\cap U)$ for all $S\in\calD$.

It is clear that the set $\sfG(\calD)$ of all games on $\calD$ is a real linear space. We will use the shorter notation $\sfG=\sfG(\calD)$ whenever $\calD$ is understood. The linear space $\sfG$ is spanned by the basis of \emph{unanimity games} $u_T$, where $\emptyset\neq T\in\calD$ and
$$
u_T(S) = \begin{cases} 1 & T\subseteq S,\\
    0 & \text{otherwise,} 
\end{cases} \qquad S\in \calD.
$$
Observe that each player $i\in P\setminus T$ is null in game $u_T$. 

The \emph{Harsanyi dividends} of a game $v\in\sfG$ are recursively defined numbers
$$
\hat{v}(S) = \begin{cases}
    0 & S = \emptyset,\\
    v(S) - \sum\limits_{\substack{T \in \mathcal{D} \\ T\subset S}} \hat{v}(T) &  S\in \calD\setminus \{\emptyset\}.
\end{cases}
$$
Equivalently, 
$$
\hat{v}(S) = \sum_{T\in \Omega(S)} (-1)^{|S|-|T|} v(T), \qquad S\in\calD,
$$
where $\Omega(S)$ is the family of all $T\in \calD$ such that $T\subseteq S$ and the order interval $\{R\in \calD \mid T\subseteq R \subseteq S\}$ is a boolean sublattice of $\calD$. The function $\hat{v}$ is a game on $\calD$ and it is also called the \emph{M\"obius transform} of $v$; see \cite{Grabisch16,Stanley2012}. The Harsanyi dividends of $v$ are the coordinates of $v$ with respect to the basis of $\sfG$ formed by the unanimity games:
\begin{equation}\label{game_linear}
    v = \sum_{\emptyset \neq T\in\calD} \hat{v}(T)\cdot u_T.
\end{equation}
Expanding the formula \eqref{game_linear} coordinatewise,
\begin{equation} \label{game_linear_coord}
    v(S) = \sum_{\substack{T\in\calD \\ T\subseteq S}} \hat{v}(T), \qquad S\in\calD. 
\end{equation}
In particular, the Harsanyi dividends of a unanimity game $u_T$ are
\begin{equation}\label{Harsanyi_unanimity}
    \hat{u}_T(S) = \begin{cases}
        1 & S=T,\\
        0 & \text{otherwise,}
       \end{cases}
       \qquad S \in \calD.
\end{equation}
We list the Harsanyi dividends for the $4$-player and $8$-player game in \ref{a:4} and \ref{a:8}, respectively.

Every null player maximal in a coalition $S$ nullifies the dividend of $S$.

\begin{proposition}\label{harsanyinull}
Let $i \in P$ be a null player in a game $v$ and $S\in\calD$. If $i \in S^*$, that is, $i$ is a maximal element in $S$, then $\hat{v}(S) = 0$. 
\end{proposition}
\begin{proof}
We will proceed by induction on the cardinality of $S$. First, assume that $|S| = |\langle i\rangle|$. In this case $S = \langle i\rangle$, and $\hat{v}(S) = v(S) - v(S \setminus i) = 0$ by the hypothesis. 
Assume that this is true for down-sets $S$ of cardinality up to $k$ and consider a down-set $S$ of cardinality $k + 1$ such that $i \in S^*$. Note that the cardinality of such a down-set is bounded by $|P \setminus (\cup_{a \in P} (i \vee a))|$, since the latter is the largest down-set (with respect to $\subseteq$) in which $i$ is maximal. So we shall assume $k \leq |P \setminus (\cup_{a \in P} (i \vee a))| - 1$. Then 
\begin{align*}
    \hat{v}(S) & = v(S) - \sum_{T \subset S} \hat{v}(T) = v(S) - \sum_{\substack{T \subset S \\ i \in T}} \hat{v}(T) - \sum_{\substack{T \subset S \\ i \notin T}} \hat{v}(T) \\
    & =
    v(S) - \sum_{\substack{T \subset S \\ i \in T^*}} \hat{v}(T) - \sum_{\substack{T \subset S \\ i \notin T}} \hat{v}(T).
\end{align*}
The last equality follows from the equivalence $i \in T \Leftrightarrow i \in T^*$ for any $T\subset S$.
By the induction hypothesis, $$\sum_{\substack{T \subset S \\ i \in T^*}} \hat{v}(T) = 0.$$ Since
\[
\sum_{\substack{T \subset S \\ i \notin T}} \hat{v}(T) = \sum_{T \subseteq S \setminus i} \hat{v}(T) = v(S \setminus i),
\]
the conclusion $\hat{v}(S) = 0$ follows.
\end{proof}

\begin{remark}\label{decomposition}
    We will comment on the iterative reasoning in the proof of Proposition \ref{harsanyinull}. For any $i \in S^*$ not necessarily null in $v$, we can expand $\hat{v}(S)$ as

\[
(v(S) - v(S \setminus i)) - \sum_{\substack{T \subset S \\ i \in T^*}} (v(T) - v(T \setminus i)) + \sum_{\substack{T \subset S \\ i \in T^*}} \sum_{\substack{U \subset T \\ i \in U^*}} (v(U) - v(U \setminus i)) - \dots 
\]
Therefore, we can decompose $\hat{v}(S)$ as the sum of marginal contributions 

\[
 \sum_{\substack{T \subseteq S \\ i \in T^*}} \beta_i(T) \cdot (v(T) - v(T \setminus i)),
\]
where each $\beta_i(T)\in\R$ depends only on the down-set $T$ and the element $i$. Note that $\beta_i(S) = 1$, for every $i \in S^*$, and $\beta_i(T) = -1$, for each $T$ with $|S|-|T| = 1$. 
\end{remark}

\section{Values}\label{sec:val}
First, we recall the notation introduced in the previous section. By $P$ we denote the boolean algebra of players, $\calD$ is the corresponding lattice of down-sets (feasible coalitions), and $\sfG$ is the real linear space of all coalitional games $v\colon \calD\to\R$. A \emph{value} on $\sfG$ is a mapping $$\phi\colon \sfG \to \R^{P}.$$ For every game $v\in\sfG$, the coordinates of vector $$\phi(v)=(\phi_i(v))_{i\in P}\in\R^{P}$$ are allocations of coalitional worth to the players $i\in P$. Any value is a possible solution concept for coalitional games. We single out usual axioms of values, which reflect both basic principles of economic rationality and mathematically convenient properties.
\begin{description}
    \item[Efficiency] $\sum\limits_{i\in P}\phi_i(v)=v(P)$, for every $v\in\sfG$. 
    \item[Positivity] $\phi_i(v)\geq 0$, for every monotone game $v\in\sfG$ and each $i\in P$.  
    \item[Carrier axiom] If $U$ is a carrier for $v\in\sfG$, then $\sum\limits_{i\in U}\phi_i(v)=v(U)$.
    \item[Null player axiom] $\phi_i(v)=0$, for any $v\in\sfG$ and each null player $i\in P$.
    \item[Symmetry] $\phi_i(v)=\phi_{\sigma(i)}(\sigma v)$ for all $v\in \sfG$, every player $i\in P$, and any boolean automorphism $\sigma\colon P\to P$, where $\sigma v(S)=v(\sigma^{-1}(S))$, $S\in\calD$.
    \item[Linearity] $\phi(\alpha v+\beta w)=\alpha\phi(v)+\beta\phi(w)$, for every $v,w\in\sfG$ and all $\alpha,\beta\in\R$. 
\end{description}

Carrier axiom is equivalent to Efficiency with Null player axiom in the standard model of coalitional games. By contrast, in our setting it is only true that Efficiency and Null player axiom imply Carrier axiom, but not conversely. Positivity is also called Monotonicity in the game-theoretic literature; cf. \cite{Weber88,Derks00}.

\subsection{Random-order values}
First, we consider the class of values obtained by averaging out the marginal vectors with respect to a probability distribution over linear extensions. This is analogous to the class of random-order values studied in the standard model of coalitional games \cite{Weber88}.

For each player $i\in P$ and every linear extension $f\in\calL(P)$, we consider the set of all players preceding player $i$ in the ranking~$f$,
$$
S_f(i)=\{j\in P\mid f(j)\leq f(i)\}.
$$
Clearly, the player $i$ is a maximal element in $S_f(i)$. Observe that $S_f(i)$ is down-set and so is the set $S_f(i)\setminus i=\{j\in P\mid f(j) < f(i)\}$.
\begin{example}
    Let $P$ be the boolean algebra of rank $3$ and the linear extension $f\in\calL(P)$ be given in Figure \ref{fig:exLE}. Then $S_f(e)=\{\bot,a,b,c,d,e\}$.
\begin{figure}[h]
    \begin{center}
\begin{tikzcd}[tips=false,column sep=1em,row sep=1.5em]
    & \top/8 \ar{dl} \ar{d} \ar{dr} & \\
    d/4 \ar{d} \ar{dr} & e/6 \ar{dl} \ar{dr} & f/7 \ar{dl} \ar{d} \\
    a/2\ar{dr} & b/3\ar{d} & c/5\ar{dl} \\
    & \bot/1 \\
\end{tikzcd}
\end{center}
\label{fig:exLE}
\caption{An example of linear extension}
\end{figure}
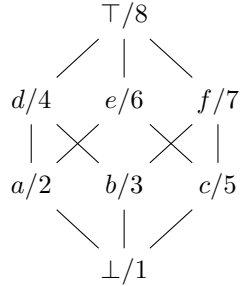
\end{example}

 The \emph{marginal contribution} of player $i\in P$ in game $v\in\sfG$ with respect to a linear extension $f\in\calL(P)$ is 
\begin{equation}\label{def:MC}
    \Delta_i^f(v) = v(S_f(i))-v(S_f(i)\setminus i).
\end{equation}
The vector $$\Delta^f(v)=(\Delta_i^f(v))_{i\in P}\in\R^P$$ is called the \emph{marginal vector}. The marginal contribution of player $i$ is the sum of Harsanyi dividends of coalitions $S\in\calD$ in which $i$ is maximal and whose players 
precede $i$ in the ranking~$f$:
\begin{equation}\label{marg_Hars}
    \Delta_i^f(v) = \sum_{\substack{S\subseteq S_f(i)\\ i\in S^*}} \hat{v}(S).
\end{equation}
The identity \eqref{marg_Hars} follows from \eqref{game_linear_coord} and from the equivalence $i\in S \Leftrightarrow i\in S^*$ valid for every down-set $S\subseteq S_f(i)$.

\begin{definition}\label{def:RO}
    Let $r$ be a probability distribution over linear extensions $\calL(P)$. The \emph{random-order value} is the value $\phi^r$ such that
    \begin{equation}\label{eq:RO}
    \phi^r_i(v) = \sum_{f\in\calL(P)} r(f)\cdot  \Delta_i^f(v), \qquad i\in P,\; v\in\sfG.
    \end{equation}
\end{definition}
Every random-order value satisfies Linearity, Efficiency, Null player axiom, and Positivity. If $r$ is the uniform distribution, $r(f)=\frac{1}{|\calL(P)|}$ for all $f\in\calL(P)$, then $\phi^r$ coincides with the hierachical value of Faigle and Kern \cite[Theorem 2]{FaigleKern92}; see Section \ref{sec:hie} for further details.

\subsection{Selectors}
We will show how random-order values relate to the family of so-called sharing values. To this end, we need the concept of selector, which was used in \cite{Derks00} to study the selectope solution of standard coalitional games. In our setting, a selector picks a maximal element from any down-set. Specifically, a \emph{selector} is a mapping $$\alpha\colon \calD\setminus \{\emptyset\}\to P$$ such that $\alpha(S)\in S^*$ for any nonempty $S\in\calD$.  Let $\calS(\calD)$ be the set of all selectors. For any  $\alpha\in\calS(\calD)$, the \emph{selector value} $\delta^{\alpha}$ is a value which gives player $i$ the sum of Harsanyi dividends of the coalitions in which $i$ is maximal according to $\alpha$:
\begin{equation}\label{def:selectorvalue}
    \delta^{\alpha}_i(v) = \sum_{\substack{S\in\calD \\ i=\alpha(S)}} \hat{v}(S), \qquad i\in P, \;v\in\sfG.
\end{equation}
 We say that a selector $\alpha\in\calS(\calD)$ is \emph{consistent} if $\alpha(S)=\alpha(T)$, for all nonempty $S,T\in\calD$ such that $S\subseteq T$ and $\alpha(T)\in S^*$. Let $\calC(\calD)$ be the set of all consistent selectors. Any selector $\alpha\in\calS(\calD)\setminus\calC(\calD)$ is called \emph{inconsistent}. We will demonstrate that consistent selectors and linear extensions are in bijection.
\begin{lemma}\label{lem:Amap}
    For every linear extension $f\in \calL(P)$, define $$\alpha_f(S) = \argmax\limits_{i\in S^*} f(i), \qquad \emptyset\neq S\in\calD.$$ Then $\alpha_f\in \calC(\calD)$ and $\delta^{\alpha_f}(v)=\Delta^f(v)$. The map sending $f\in\calL(P)$ to $\alpha_f\in \calC(\calD)$ is injective.
\end{lemma}
\begin{proof}
    By the definition of $\alpha_f$, we obtain $\alpha_f(S)\in S^*$ for any nonempty $S\in\calD$. If $S\subseteq T$ and $\alpha_f(T)\in S^*$, then $\alpha_f(T)$ is the maximizer of $f$ over $S^*$. Therefore, $\alpha_f(S)=\alpha_f(T)$, and $\alpha_f\in \calC(\calD)$. 

    We prove the identity $\delta^{\alpha_f}_i(v)=\Delta^f_i(v)$ for each $i\in P$. This is true by \eqref{marg_Hars} and by the equivalence of these three conditions for any $S\in\calD$:
    $$
    i=\alpha_f(S) \qquad \Leftrightarrow \qquad f(j)\leq f(i), \;i\in S^*,\; \forall j\in S \quad \Leftrightarrow \quad S\subseteq S_f(i),\; i\in S^*.
    $$
    Suppose $f_1,f_2\in\calL(P)$ and $f_1\neq f_2$. Then, by \eqref{equality_crit}, there exist necessarily imcomparable elements $i,j\in P$ such that $f_1(i)>f_1(j)$ and $f_2(i) < f_2(j)$. Define $S$ to be the down-set whose set of maximal elements is precisely $\{i,j\}$. Then $\alpha_{f_1}(S)=i\neq j=\alpha_{f_2}(S)$.
\end{proof}

Let $\alpha$ be a consistent selector. We define a function $g_\alpha\colon \{1,\dots,|P|\}\to P$ recursively. Let $g_\alpha(|P|)=\alpha(P)=\top$ and 
\begin{equation*}
    g_\alpha(\ell)=\alpha (P\setminus \{g_\alpha(|P|),\dots,g_\alpha(\ell + 1)\}), \qquad \ell=1,\dots, |P|-1.
\end{equation*}
By the construction, $g_\alpha$ is bijective. 
The bijection $g_\alpha$ selects and removes the elements from $P$ (starting from the top element $\top$) in such a way that the removed element is maximal in the set from which it is removed, and thus it generates a descending chain of down-sets. Define 
\begin{equation}\label{def:map}
    f_\alpha=g_\alpha^{-1}.
\end{equation}

\begin{lemma}
    Let $\alpha\in\calC(\calD)$. Then $f_\alpha\in\calL(P)$.
\end{lemma}

\begin{proof}
    Clearly, $f_\alpha$ given by \eqref{def:map} is bijective. We need to prove that $f_\alpha$ is order-preserving. Suppose that $i\preceq j$. There are $\ell,\ell' \in \{1,\dots,|P|\}$ such that $i=g_\alpha(\ell)$ and $j=g_\alpha(\ell')$. Put
   \[
        S_\ell  = P\setminus \{g_\alpha(|P|),\dots,g_\alpha(\ell + 1)\}\quad \text{and}\quad
        S_{\ell'} = P\setminus \{g_\alpha(|P|),\dots,g_\alpha(\ell' + 1)\}.
    \] 
    Then $g_\alpha(\ell)=\alpha (S_\ell) \preceq \alpha (S_{\ell'})=g_\alpha(\ell')$.
    Assume that $\ell > \ell'$. Then $S_{\ell'} \subset S_{\ell}$. However, $\alpha(S_\ell) \preceq \alpha(S_{\ell'})$, which contradicts the maximality of $\alpha(S_\ell)$ in $S_\ell$. Therefore, we  obtain $\ell \leq \ell'$, so $f_\alpha\in \calL(P)$.
\end{proof}

The next lemma shows that there is a one-to-one correspondence between linear extensions and consistent selectors. 
\begin{lemma}\label{lem:MI}
The maps $A\colon \calL(P) \to \calC(\calD)$ and $F\colon\calC(\calD) \to \calL(P)$ given by $$A(f)=\alpha_f\qquad \text{and} \qquad F(\alpha)=f_\alpha,$$ respectively, are mutually inverse.
\end{lemma}

\begin{proof}
 The map $A$ is injective by Lemma~\ref{lem:Amap}, so we only need to prove that $A \circ F$  is the identity  on $\calC(\calD)$. 
Let $\alpha \in \mathcal{C}(\mathcal{D})$. We shall prove that $\alpha_{f_{\alpha}} =A(F(\alpha))= \alpha$. 
Let $\emptyset\ne S \in \mathcal{D}$ and assume that $\alpha_{f_{\alpha}}(S) = s$. This assumption is equivalent to any of the three conditions:
$$
\argmax\limits_{i \in S^*} f_{\alpha}(i) = s\quad \Leftrightarrow\quad f_{\alpha}(s) \geq f_{\alpha}(i),\;\forall i \in S\quad \Leftrightarrow\quad g_{\alpha}^{-1} (i) \leq g_{\alpha}^{- 1}(s), \;\forall i \in S.
$$
Let $l_S$ be the maximal integer from $\{1, \dots, |P|\}$ such that $g_{\alpha}(l_S) \in S$ and let $K = P \setminus \{g_{\alpha}(|P|), g_{\alpha}(|P| - 1), \dots, g_{\alpha}(l_S +1)\}$. This implies that $K$ is the maximal down-set determined by $g_{\alpha}$ such that $\alpha(K) \in S$. Assume that $\alpha(K) = s' \neq s$. By maximality of $l_S$, we obtain $g_{\alpha}^{-1}(s') = l_S \geq g_{\alpha}^{-1}(i)$ for all $i\in S$. In particular, we have $g_{\alpha}^{-1}(s') > g_{\alpha}^{-1}(s)$, which contradicts the hypothesis. Since $S \subseteq K$ and $\alpha$ is consistent, we deduce that $\alpha(S)=s$. Assume now that $\alpha(S) = s$ and let $l_S$ and $K$ be as above. By the consistency of $\alpha$, we have that $\alpha(K) = s$ and $g_{\alpha}^{-1}(s) = l_S$. Since $l_S$ is maximal such that $g_{\alpha}(l_S) \in S$, we have $g_{\alpha}^{-1}(i) \leq g_{\alpha}^{-1}(s)$ for all $i \in S$.
\end{proof}

Any consistent selector value \eqref{def:selectorvalue} is equal to the marginal vector associated with the ranking \eqref{def:map}.
\begin{lemma}\label{lem:consel}
For every consistent selector $\alpha\in\calC(\calD)$ and all $v\in\sfG$,
        \begin{equation}\label{eq:id}
            \delta^{\alpha}(v)=\Delta^{f_\alpha}(v).
        \end{equation}
\end{lemma}

\begin{proof}
     Let $\alpha\in\calC(\calD)$. For the sake of compact notation in the proof, put $f=f_\alpha$ and $g=g_\alpha$. Let $v\in\sfG$, $i\in P$ and assume that $i=g(\ell)$. Then \eqref{marg_Hars} implies
    \begin{equation}\label{eq:marg}
        \Delta_{g(\ell)}^f(v) = \sum_{\substack{S\subseteq S_f(g(\ell))\\ g(\ell)\in S^*}} \hat{v}(S).
    \end{equation}
By the definition of selector value,
    \begin{equation}\label{eq:sel}
        \delta_{g(\ell)}^{\alpha}(v) = \sum_{\substack{S\in\calD\\ \alpha(S)=g(\ell)}} \hat{v}(S).
    \end{equation}
    If we show that 
    \begin{equation}\label{eqsets}
        \{S\in\calD \mid S\subseteq S_f(g(\ell)),\; g(\ell)\in S^*\} = \{S\in\calD \mid \alpha(S)=g(\ell)\},
    \end{equation}
    then \eqref{eq:marg} and \eqref{eq:sel} are equal. First, observe that
    $$
    S_f(g(\ell)) = \{j\in P \mid g^{-1}(j)\leq \ell\} = \{g(1),\dots, g(\ell)\}.
    $$
    Let $S\subseteq S_f(g(\ell))$ and $g(\ell)\in S^*$. Then $\alpha(S_f(g(\ell)))=g(\ell)$ by the definition of $g$. As $\alpha$ is consistent, $\alpha(S)=g(\ell)$. Conversely, let $S\in\calD$ be such that $\alpha(S)=g(\ell)$. Then necessarily $g(\ell)\in S^*$. We want to prove $S\subseteq S_f(g(\ell))$. There exists some $j\in S^*$ such that $f(k)\leq f(j)$ for all $k\in S$. Then 
    \begin{equation}\label{subsetS}
        S\subseteq S_f(j).
    \end{equation}
     Since we can write $j=g(\ell')$ for some $\ell'$, we get $\alpha(S_f(j))=j$ by the definition of~$g$. By consistency of $\alpha$, it follows that $j=\alpha(S)=g(\ell)$. Therefore $S\subseteq S_f(g(\ell))$ by \eqref{subsetS}. In conclusion, \eqref{eqsets} is true, which finishes the proof of \eqref{eq:id}.
    \end{proof}
Lemma \ref{lem:consel} ensures this property: If $\alpha$ is a consistent selector and $v$ is a monotone game, then the vector $\delta^{\alpha}(v)$ has nonnegative coordinates. This is not true in case $\alpha$ is inconsistent.

\begin{example}
Let $P$ be the boolean algebra of rank $3$ -- see Figure \ref{FigBA8}. Let $T =\langle a,b,c\rangle$ and $S=\langle a,b\rangle$ and consider any inconsistent selector $\alpha$ such that $\alpha(T) = a$, $\alpha(S) = b$, and $\alpha(\langle a,c\rangle)=c$. Then $\delta^{\alpha}_a (v) = \hat{v}(\langle a,b,c\rangle ) + \hat{v}(\langle a\rangle)$, which is not necessarily positive even if $v$ is monotone.
\end{example}

\subsection{Sharing values}
We introduce sharing systems analogously to \cite[Section 3.5]{Grabisch16} or \cite{Derks00}.

\begin{definition}
    A \emph{sharing system} is a mapping $q\colon \calD\times P\to[0,1]$ such that for each nonempty coalition $S\in\calD$,
\begin{itemize}
    \item $q(S,.)\colon P\to [0,1]$ is a probability distribution,
    \item  $q(S,i)=0$ for each $i\in P\setminus S^*$, and
    \item $q(\emptyset,i)=0$ for all $i\in P$.
\end{itemize}
   The \emph{sharing value} (or \emph{Harsanyi solution}) associated with a sharing system $q$ is a value $\pi^{q}$ defined by 
    \begin{equation}\label{def:sharingv}
    \pi^q_i(v) = \sum_{\substack{S\in\calD \\ i\in S^*}} q(S,i)\cdot \hat{v}(S), \qquad i\in P,\; v\in\sfG.
    \end{equation}
\end{definition}

When the coalitional worth is distributed according to a sharing value, every dividend $\hat{v}(S)$ is split only among the maximal players in $S$. We argue that this restriction to the distribution among the maximal players is sensible in our setting:
\begin{itemize}
    \item The maximal players $i\in S^*$ can be thought of as the superiors of players in $S\setminus S^*$, so it makes sense to limit the allocation process of the Harsanyi dividend $\hat{v}(S)$ to them.
    \item The set $S^*$ is an antichain, so we can allocate the Harsanyi dividend $\hat{v}(S)$ freely as long as we respect the sharing system $q$.
    \item Since the player set is a boolean algebra, every player acts the same number of times as a maximal player.
    \item It is computationally simpler to evaluate the sum \eqref{def:sharingv} for a smaller number of players with the size of down-sets $S$ increasing.

\end{itemize}

\begin{lemma}\label{lem:sv}
    Every sharing value $\pi^q$ fullfills Linearity, Efficiency, and Null player axiom.
\end{lemma}
\begin{proof}
    Linearity of $\pi^q$ follows immediately from the linearity of Harsanyi dividends: $\widehat{v_1+v_2}=\widehat{v_1}+\widehat{v_2}$ and $\widehat{\alpha v_1}=\alpha\widehat{v_1}$, for all $v_1,v_2\in\sfG$, $\alpha\in\R$. As for Efficiency, let $v\in\sfG$. Then
    \begin{align*}
        \sum_{i\in P}  \pi^q_i(v) & = \sum_{i\in P} \sum_{\substack{S\in\calD \\ i\in S^*}} q(S,i)\cdot \hat{v}(S) = \sum_{S\in\calD} \sum_{\substack{i\in S^*}} q(S,i)\cdot \hat{v}(S)\\
        & = \sum_{S\in\calD} \hat{v}(S) \underbrace{\sum_{\substack{i\in S^*}} q(S,i)}_1 = \sum_{S\in\calD} \hat{v}(S) = v(P),
    \end{align*}
    where the last equality is a consequence of \eqref{game_linear_coord}. Hence Efficiency of $\pi^q$. 
    Finally, let $i\in P$ be a null player in $v$. Then every Harsanyi dividend $\hat{v}(S)=0$ whenever $i\in S^*$ by Proposition \ref{harsanyinull}, which implies that $\pi_i^q=0$. 
\end{proof}
A sharing value may not be positive as the following examples demonstrate.
\begin{example}[Priority value]\label{ex:nonpos}
    We adopt the concept of priority value constructed for the different class of games \cite[Formula (4)]{Beal21} (see also the discussion in Section \ref{sec:other}) and adapt the definition to our setting. This leads to the sharing system $q$ such that $$q(S,i)=
    \begin{cases} \frac{1}{|S^*|} & i \in S^*,\\
        0 & \text{otherwise,}
    \end{cases} 
    $$ where the associated sharing value $\pi^q$ given by \eqref{def:sharingv} is 
\begin{equation}\label{sharingD:prio}
\pi^q_i(v)=\sum_{\substack{S\in\calD \\ i\in S^*}} \frac{\hat{v}(S)}{|S^*|}, \qquad v\in\sfG,\;i\in P.
\end{equation}
The idea is that each Harsanyi dividend is distributed among the maximal players uniformly. However, the sharing value \eqref{sharingD:prio} fails Positivity axiom. An example of a monotone game $v$ for which some $\pi_i^q(v)$ is negative can be found using Table \ref{tab:H8} of Harsanyi dividends.
\end{example}

\begin{example}[Proportional value] As a more sophisticated variant to the uniform split used in the priority value (Example \ref{ex:nonpos}), we can distribute the dividend proportionally to each player's rank $\rho(i)$ in the boolean algebra $(P,\preceq)$. We define
$$
\rho(S^*) = \sum_{i\in S^*} \rho(i), \qquad S\in\calD,
$$ 
and $q(S,i)= \frac{\rho(i)}{\rho(S^*)}$, for any $i\in S^*$, and $q(S,i)=0$, otherwise. Then $q$ is a sharing system and the \emph{proportional value} is the associated sharing value $\pi^q$, that is,
$$
\pi^q_i(v) = \sum_{\substack{S\in\calD \\ i\in S^*}} \frac{\rho(i)}{\rho(S^*)}\cdot \hat{v}(S), \qquad i\in P, \;v\in\sfG.
$$
It follows from \eqref{Harsanyi_unanimity} and from the definition of proportional value that
\begin{equation}\label{prop_unanimity}
    \pi_i(u_T) = \begin{cases} 
        \frac{\rho(i)}{\rho(T^*)} & i \in T^*,\\
        0 & \text{otherwise.}
    \end{cases}
\end{equation}

It can be shown that the proportional values satisfies Symmetry. First, we note that for any nonempty $T\in\calD$ and any boolean automorphism $\sigma$,
    \begin{equation}\label{eq_un_switch}
        \sigma u_T=u_{\sigma(T)}.
    \end{equation}
    Indeed, this follows from the equivalence of the inclusions $T\subseteq \sigma^{-1}(S)$ and $\sigma(T)\subseteq S$, for any $S\in\calD$. We will show that
    \begin{equation}\label{unan_sym}
        \pi^q_{\sigma(i)}(\sigma u_T) = \pi^q_i(u_T), \qquad i\in P.
    \end{equation}
    Using \eqref{eq_un_switch} and \eqref{prop_unanimity}, we obtain
    $$
    \pi^q_{\sigma(i)}(\sigma u_T) = \pi^q_{\sigma(i)}(u_{\sigma(T)})=\frac{\rho(\sigma(i))}{\sum\limits_{j\in \sigma(T)^*} \rho(j)} 
    $$
    when $\sigma(i)\in\sigma(T)^*$. The key observation is that any boolean automorphism $\sigma\colon P\to P$ preserves the ranks of players in $P$, that is, $\rho(\sigma(i))=\rho(i)$ for all $i\in P$. Moreover, any such $\sigma$ is necessarily an order-preserving map. Thus, the condition $\sigma(i)\in\sigma(T)^*$ is equivalent to $i\in T^*$, and
    $$
    \frac{\rho(\sigma(i))}{\sum\limits_{j\in \sigma(T)^*} \rho(j)} = \frac{\rho(i)}{\sum\limits_{j\in \sigma(T^*)} \rho(j)} = \frac{\rho(i)}{\sum\limits_{k\in T^*} \rho(\sigma(k))} =  \frac{\rho(i)}{\sum\limits_{k\in T^*} \rho(k)}.
    $$
    This implies that \eqref{unan_sym} holds, so $\pi^q$ is symmetric over all unanimity games $u_T$.

    Now, let $v\in\sfG$ be an arbitrary game, and consider any player $i\in P$ and any boolean automorphism $\sigma$. Then, by \eqref{game_linear}, linearity of $\pi^q$ and the identity $\widehat{\sigma v}=\sigma\hat{v}$, we get:
    $$
    \pi^q_{\sigma(i)}(\sigma v) = \sum_{\emptyset \neq T\in\calD} \widehat{\sigma v}(T) \pi^q_{\sigma(i)}(u_T)=\sum_{\emptyset \neq T\in\calD} \hat{v}(\sigma^{-1}(T))\pi^q_{\sigma(i)}(u_T).
    $$
    Employing \eqref{unan_sym}, the last term on the right-hand side is equal to
    $$
    \sum_{\emptyset \neq T\in\calD} \hat{v}(T)\pi^q_{\sigma(i)}(u_{\sigma(T)}) = \sum_{\emptyset \neq T\in\calD} \hat{v}(T) \pi^q_i(u_T) = \pi^q_i(v).
    $$

    However, also the proportional value fails Positivity. Th smallest counterexample can be exhibited on the boolean algebra of rank $4$, which involves $2^4$ players and $168$ coalitions (see Table \ref{tabDed}).
\end{example}

It turns out that the three axioms from Lemma \ref{lem:sv} and Positivity characterize the class of positive sharing values. This result extends \cite[Theorem 4(a)]{Derks00}. 
\begin{proposition}\label{pro:PSV}
Let $\phi \colon \mathcal{G} \to \mathbb{R}^P$ a positive value. The following are equivalent.
\begin{enumerate}
    \item $\phi$ is a sharing value. 
    \item $\phi$ satisfies Linearity, Null player axiom, and Efficiency.
\end{enumerate}
\end{proposition}

\begin{proof}
   The implication $1. \Rightarrow 2.$ is Lemma \ref{lem:sv}.  Conversely, let $\phi$ be a positive value satisfying the three axioms. Any game $v$ can be re\-presented as a linear combination of unanimity games \eqref{game_linear}. Then linearity of $\phi$ gives
\[
\phi_i (v) = \sum_{\emptyset \neq S\in \mathcal{D} } \hat{v} (S) \phi_i (u_S), \qquad i\in P.
\]
We will prove that $q \colon \calD \times P \to [0,1]$ such that $q(S, i)=\phi_i (u_S)$ is a sha\-ring system. Since $\phi$ satisfies Positivity, clearly $\phi_i (u_S) \geq 0$. By Efficiency, $$\sum_{i \in P} q(S,i)=\sum_{i \in P} \phi_i (u_S) = u_S (P) = 1.$$ It is easy to see that any player $i \in P \setminus S^*$ is a null player in $u_S$, so by Null player axiom  $\phi_i (u_S) = 0$ for every $i \in P \setminus S^*$. Consequently, $q$ is a sharing system and $\phi$ coincides with the sharing value $\pi^q$. 
\end{proof}

Our next goal is to give an alternative description of positive sharing values and to characterize the random-order values within positive sharing values. To this end, we will show that the class of sharing values coincides with the class of average selector values. This connection becomes instrumental in studying positivity of the sharing values \eqref{eq:selq}. The following result is essentially \cite[Lemma 4]{Derks00}, which was proved in the context of standard coalitional games.

\begin{lemma} \label{pro:sharing}
    The following assertions hold true.
    \begin{enumerate}
        \item For any sharing system $q$, this function is a probability distribution:
        $$
        p_q(\alpha) = \prod_{\emptyset\neq S\in\calD} q(S,\alpha(S)), \qquad \alpha\in\calS(\calD).
        $$
        The sharing value corresponding to $q$ is
        \begin{equation}\label{eq:sharesel}
            \pi^q(v) = \sum_{\alpha \in \calS(\calD)} p_q(\alpha) \cdot \delta^\alpha(v), \qquad v\in\sfG.
        \end{equation}
        \item Conversely, let $p$ be a probability distribution over $\calS(\calD)$. Define $$q_p(S,i)=\sum_{\substack{\alpha \in \calS(\calD)\\ \alpha(S)=i}} p(\alpha), \qquad \emptyset\neq S\in\calD,\; i\in S^*,$$
        and $q_p(S,i)=0$, otherwise. Then $q_p$ is a sharing system and
        \begin{equation}\label{eq:selq}
            \pi^{q_p}(v) = \sum_{\alpha\in\calS(\calD)} p(\alpha)\cdot \delta^\alpha(v), \qquad v\in\sfG.
        \end{equation}
    \end{enumerate}
\end{lemma}

\begin{proof}
    1. That $p_q$ is a probability distribution follows from an easy adaptation of the second part of the proof. Since both values $\pi^q(v)$ and $\delta^\alpha(v)$ are linear maps in games $v\in\sfG$, it suffices to verify \eqref{eq:sharesel} for any unanimity game $u_T$. We have
    $$
    \pi^q_i(u_T) = \begin{cases} q(T,i) & i\in T^*,\\ 
        0 & i\notin T^*,
    \end{cases}
    \qquad
    \text{and}
    \qquad
    \delta_i^\alpha(u_T)=\begin{cases}
        1 & i=\alpha(T), \\
        0 & i\neq \alpha(T).
    \end{cases}
    $$
    If $i\notin T^*$, then $i\neq \alpha(T)$ and both sides of \eqref{eq:sharesel} are zero. Let $i\in T^*$. Then \eqref{eq:sharesel} reads as
    \begin{equation}\label{eq:toprove}
        q(T,i) = \sum_{\substack{\alpha\in\calS(\calD) \\ \alpha(T)=i}} \prod_{\emptyset\neq S\in\calD} q(S,\alpha(S)).
    \end{equation}
    Let $T_1\neq T,\emptyset$. Then the sum on the right-hand side above is equal to
    \begin{equation*}
        q(T,i)\sum_{\substack{\alpha\in\calS(\calD) \\ \alpha(T)=i}} \prod_{\substack{S\in\calD \\ S\neq \emptyset,T}} q(S,\alpha(S))=q(T,i)\sum_{\substack{\alpha\in\calS(\calD) \\ \alpha(T)=i}}q(T_1,\alpha(T_1)) \prod_{\substack{S\in\calD \\ S\neq \emptyset,T,T_1}} q(S,\alpha(S)).
    \end{equation*}
    The last expression can be written as 
    \begin{equation} \label{eq:last}
    q(T,i)\underbrace{\sum_{j\in T_1^*} q(T_1,j)}_1\sum_{\substack{\alpha\in\calS(\calD)\\ \alpha(T)=i \\ \alpha(T_1)=j}}\prod_{\substack{S\in\calD \\ S\neq \emptyset,T,T_1}} q(S,\alpha(S)).
    \end{equation}
    By repeating the last step for the remaining down-sets in $\calD\setminus \{\emptyset,T,T_1\}$, the term \eqref{eq:last} is reduced to $q(T,i)$. This proves identity \eqref{eq:toprove}.

    2. The mapping $q_p$ is a sharing system since, for every nonempty $S\in\calD$,
    $$
    \sum_{i\in S^*} q_p(S,i)= \sum_{i\in S^*} \sum_{\substack{\alpha \in \calS(\calD)\\ \alpha(S)=i}} p(\alpha)=\sum_{\alpha\in\calS(\calD)} p(\alpha)=1.
    $$ 
   The second equality above is true as the sets $\{\alpha \in \calS(\calD) \mid \alpha(S)=i\}$ form a partition of $\calS(\calD)$. In order to show that \eqref{eq:selq} holds true, let $v=u_T$ for nonempty $T\in\calD$. Suppose $i\in T^*$. Then
    $$\pi_i^q(u_T)=q_p(T,i)=\sum_{\substack{\alpha \in \calS(\calD)\\ \alpha(T)=i}} p(\alpha)= \sum_{\alpha\in\calS(\calD)} p(\alpha) \delta_i^\alpha(u_T).$$
   Let $i\notin T^*$. Then both sides of \eqref{eq:selq} are zero, as $\pi_i^q(u_T)=0$ and $\delta_i^\alpha(u_T)=0$ for any selector $\alpha$.
\end{proof}

We will use Lemma \ref{pro:sharing} to characterize random-order values in terms of certain sharing values. Namely, the family of random-order values coincides with the family of sharing values $\pi^{q_p}$ (see \eqref{eq:selq}) whose sharing systems $q_p$ correspond to probability distributions $p$ supported by a subset of consistent selectors.

\begin{proposition} \label{random sharing}
Random-order values are exactly the sharing values \eqref{eq:selq} such that $p(\alpha) = 0$ for every $\alpha\in\calS(\calD)\setminus\calC(\calD)$. 
\end{proposition}

\begin{proof}
Let $\phi^r$ be a random-order value \eqref{eq:RO}. Since $r$ is a probability distribution over linear extensions, by the one-to-one correspondence between linear extensions in $\calL(\calP)$ and consistent selectors $\calC(\calD)$ (see Lemma \ref{lem:MI}), we can define the probability distribution $\bar{r}$ over all selectors by
$$
\bar{r}(\alpha)=
\begin{cases}
    r(f_\alpha) & \alpha\in\calC(\calD), \\
    0 & \alpha \in \calS(\calD)\setminus \calC(\calD),
\end{cases}
$$
where $f_\alpha$ is the linear extension \eqref{def:map} corresponding to $\alpha$.
Using this correspondence and \eqref{eq:id}, we obtain for every $i\in P$ and every $v\in\sfG$,
\[
 \phi^r_i(v) = \sum_{f\in\mathcal{L}(P)} r(f)  \Delta_i^f(v) = \sum_{f\in\mathcal{L}(P)} r(f)  \delta_i^{\alpha_f}(v) = 
\sum_{\alpha \in \mathcal{S}(\mathcal{D})} \bar{r}(\alpha)  \delta_i^{\alpha}(v).
\]
By the second part of Lemma \ref{pro:sharing}, if $q_{\bar{r}}$ is the sharing system associated with~$\bar{r}$, then the last sum is equal to the sharing value $\pi_i^{q_{\bar{r}}}(v)$.

Conversely, let $\pi^{q_p}$ be the sharing value \eqref{eq:selq} such that $p(\alpha) = 0$ for every $\alpha\in\calS(\calD)\setminus\calS(\calD)$. Then it is easy to see using the above equalities and Lemma~\ref{pro:sharing}  that $\pi^{q_p}$ is a random-order value.
\end{proof}

In particular, Proposition \ref{random sharing} and the positivity of random-order values imply that every sharing value satisfying the property from Proposition \ref{random sharing} is positive. The natural question of interest is whether the converse holds. Namely, if a probability distribution $p$ over $\calS(\calD)$ gives positive probability $p(\alpha)>0$ to some inconsistent selector $\alpha\in\calS(\calD)\setminus \calC(\calD)$, is the resulting sharing value automatically non-positive? We shall prove that this is not the case by providing a counterexample --- see Example \ref{positivesharing}. To this end, we introduce the concept of ``local'' inconsistency. 

\begin{definition}\label{def:inco}
We say that a selector $\alpha$ is \emph{inconsistent on} $t\in P$ if there exist down-sets $S,T\in\calD$ such that $S \subset T$, $\alpha(T) = t \in S^*$, and $\alpha(S) \neq \alpha(T)$. Otherwise, we shall say that a selector $\alpha$ is \emph{consistent on} $t$. 
\end{definition}
\noindent
Note that a consistent selector $\alpha\in\calC(\calD)$ is consistent on every element $t\in P$ and that an inconsistent selector $\alpha\in\calS(\calD)\setminus\calC(\calD)$ may be consistent on some elements.

\begin{lemma}\label{consistent}
Let $\alpha$ be a selector consistent on $t\in P$. Assume that there exists $T\in\calD$ such that for every $S\in\calD$ we have $\alpha(S)=\alpha(T)=t$ and $S \subseteq T$. Then the corresponding selector value of game $v\in\sfG$ is 
$\delta_t^{\alpha}(v) = v(T) - v(T \setminus t)$. 
\end{lemma}

\begin{proof}
By the consistency of $\alpha$ on $t$, we obtain:
\begin{align*}
    \delta_t^{\alpha}(v) & = \sum_{\substack{S \in \calD \\ t=\alpha(S)}} \hat{v}(S) = \hat{v}(T) +  \sum_{\substack{S \in \calD \\ t\in S \subset T}} \hat{v}(S)=v(T) - \sum_{\substack{S \in \calD \\ S \subset T}} \hat{v}(S) + \sum_{\substack{S \in \calD \\ t\in S \subset T}} \hat{v}(S) \\
    & = v(T) - \sum_{\substack{S \in \calD \\ t\notin S \subset T}} \hat{v}(S) = v(T) - \sum_{\substack{S \in \calD \\ S \subseteq T \setminus t}} \hat{v}(S) =
    v(T) - v(T \setminus t).\qedhere
\end{align*}
\end{proof}
\noindent
Observe that any consistent selector $\alpha$ necessarily satisfies the assumption of lemma above since  $T=\bigcup\{D \in \calD \mid \alpha(D) = t\}$. Also, Lemma~\ref{consistent} implies that $\delta_t^{\alpha}(v) \geq 0$ for any selector $\alpha$ satisfying the assumption and for any monotone game $v\in\sfG$.

The next example presents a positive sharing value with $p$ positive on some inconsistent selectors.

\begin{example}\label{positivesharing}
    We consider the boolean algebra $P$ of rank $3$ as in Figure \ref{FigBA8}. Consider arbitrary selectors $\alpha,\beta\in\calS(\calD)$ such that:
\begin{enumerate}
    \item $\alpha$ is consistent and satisfies $\alpha(\langle a,b,c\rangle  )=c$, $\alpha(\langle a,b\rangle  )=a$, $\alpha(\langle b,e\rangle  )=e$, $\alpha(\langle e,f\rangle  )=f$, $\alpha(\langle d,e,f\rangle  )=d$.
    \item $\beta$ satisfies $\beta(\langle a,b,c\rangle  ) = \beta(\langle a,c\rangle  )=a$, $\beta(\langle a,b\rangle  ) = b$, $\beta(\langle b,c\rangle  ) = c$, and $\beta$~coincides with $\alpha$ on all the other down-sets. Note that $\beta$ is inconsistent. 
\end{enumerate}
Let $p$ be an arbitrary probability distribution over selectors such that $p(\beta) > 0$, $p(\alpha) = 1 - p(\beta)$, and the probability of any other selector is zero. Note that  $\delta_i^{\alpha}(v),\delta_i^{\beta}(v) \ge 0$ for  $i \in\{ d, e, f\}$. Indeed, both selectors are consistent on those elements and they satisfy the hypothesis of Lemma \ref{consistent}, hence the conclusion. We shall proceed with the calculations of $\delta_i^{\beta}(v)$ for $i\in \{a, b, c\}$ and a monotone game $v\in\sfG$:
\begin{align*}
    \delta_a^{\beta}(v) & = v(\langle a,b,c\rangle  ) - v(\langle a,b\rangle  ) - v(\langle b,c\rangle  ) + v(\langle b\rangle  ) + v(\langle a\rangle  ), \\
    \delta_b^{\beta}(v) & = v(\langle a,b\rangle  ) - v(\langle a\rangle  )\ge 0, \\
    \delta_c^{\beta}(v)& = v(\langle b,c\rangle  ) - v(\langle b\rangle  )\ge 0.
\end{align*}
The only interesting case is the sharing value \eqref{eq:selq} of player $a$: 
\[
\pi_a^{q_p}(v)= p(\beta)\cdot  \delta_a^{\beta}(v) + (1 -p(\beta))\cdot  \delta_a^{\alpha}(v),
\]
where $\delta_a^{\alpha}(v)=v(\langle a,b\rangle  ) - v(\langle b\rangle  )$. Note that if $v(\langle a,b\rangle  ) = v(\langle b\rangle )$, then the sharing value is automatically positive. 
Therefore, assume that $v(\langle a,b\rangle  ) - v(\langle b\rangle )$ is strictly positive. Then $\delta_a^{\beta}(v)\ge v(\langle b\rangle )-v(\langle a,b\rangle  ) $, which gives 
\begin{align*}
\pi_a^{q_p}(v) & \ge p(\beta)\cdot (v(\langle b\rangle )-v(\langle a,b\rangle  )) + (1-p(\beta))\cdot (v(\langle a,b\rangle  ) - v(\langle b\rangle )) \\
& = (1-2p(\beta))\cdot (v(\langle a,b\rangle  ) - v(\langle b\rangle ))
\end{align*}
In particular, the sharing value $\pi_a^{q_p}(v)$ is positive whenever $p(\beta) < \tfrac{1}{2}$.
\end{example}
We have already seen (Example \ref{ex:nonpos}) that a sharing value may fail to be positive. This result provides another perspective at the lack of positivity.
\begin{proposition}\label{ex:svneg}
    Let $p$ be a probability distribution such that $p(\alpha)>\tfrac 12$ for a unique inconsistent selector $\alpha$.  Then there exists a monotone game $v$ such that $\pi_t^{q_p}(v)< 0$ for some player $t$.
\end{proposition}

\begin{proof}
Let $\alpha$ be a selector inconsistent on $t\in P$. We can assume that there exists $T\in\calD$ with $\alpha(T) = t$ and for any $S\in\calD$ such that $S \subseteq T$, we have $\alpha(S) = t$. To see this, consider the down-set $T'=\bigcup\{D \in \calD \mid \alpha(D) = t\}$. If $\alpha(T')=t$, we can take $T'$ as $T$, otherwise $\alpha$ is inconsistent also on $t'=\alpha(T')$ and we can iterate the argument considering the down-set $\bigcup\{D' \in \calD \mid \alpha(D') = t'\}$. With $T$ and $t$ as above and with the convention that the sum are over $S\in\calD$, we get

\begin{align*}
    \delta_t^{\alpha}(v) & = \sum_{\substack{t=\alpha(S)}} \hat{v}(S) = \hat{v}(T) +  \sum_{\substack{S \subset T \\ \alpha(S)=t}} \hat{v}(S)=v(T) - \sum_{\substack{S \subset T}} \hat{v}(S) + \sum_{\substack{S \subset T \\ \alpha(S)=t}} \hat{v}(S)  \\
    & = v(T) - \sum_{\substack{S \subset T \\ \alpha(S) \neq t}} \hat{v}(S) =
 v(T) - \sum_{\substack{S \subseteq T \setminus t}} \hat{v}(S) -  \sum_{\substack{t\in S \subset T \\ \alpha(S) \neq t}} \hat{v}(S) \\ & =
 v(T) - v(T \setminus t) - \sum_{\substack{t\in S \subset T \\ \alpha(S) \neq t}} \hat{v}(S).
\end{align*}
Let $S_1, \dots, S_k$ be the maximal elements of $\{S \in \calD \mid S \subset T,\, t \in S, \, \alpha(S) \neq t\}$ and consider any monotone game $v$ which satisfies the following assumptions:

\begin{itemize}
\item $v(D) = 0$ for every $D \subset S_i$ for some $i \in \{1, \dots, k\}$.
\item $v (D) - v(D \setminus t)=0$ for every $D$ such that $t \in D^*$ and $D \notin \{S_1, \dots, S_k\}$.
\item $v(S_i)>0$ for every $i \in \{1, \dots, k\}$. 
\end{itemize}
Then $\delta_t^{\alpha}(v) = -kv(S_i)$ and, by Lemma \ref{consistent}, $\delta_t^{\beta}(v)=0$ or  $\delta_t^{\beta}(v)=v(S_i)$ for any consistent selector $\beta$. Note that, given the consistent selectors $\beta_1, \dots, \beta_n$ whose probability is positive, we have $p(\alpha) > \sum_{i=1}^{n} p(\beta_i) \eqqcolon p'$. In either case, $\pi_t^{q_p}(v) \leq p'v(S_i) - (1-p')kv(S_i) = v(S_i)(p' - (1-p')k)$, which is strictly negative since $p' < (1-p')k$. 
\end{proof}

On the one hand, Example \ref{positivesharing} shows that a sharing value \eqref{eq:selq} may be positive although there exists some inconsistent selector with a strictly positive probability. On the other, Proposition \ref{ex:svneg} shows that in case of a unique inconsistent selector with positive probability greater than $\tfrac 12$, there exists a monotone game $v$ for which the corresponding sharing value is not positive.

The question whether it is possible to characterize positive sharing values \eqref{eq:selq} by the probability distributions $p$ over selectors remains open for further research. In particular, the results above lead to the following conjecture:
 If a sharing value \eqref{eq:selq} is positive, then  $$\sum_{\alpha \in\calS(\calD)\setminus \calC(\calD)} p(\alpha) \le  \sum_{\alpha \in\calC(\calD)} p(\alpha).$$ In summary, we have obtained the strict inclusions
\begin{center}
random-order values $\subset$ positive sharing values $\subset$ sharing values. 
\end{center}

By Proposition \ref{pro:PSV}, Proposition \ref{random sharing}, and Example \ref{harsanyinull}, we know that random-order values are strictly contained in the family of values satisfying Linearity, Positivity, and Null player axiom. Such values have necessarily the form \eqref{eq:valform}.

\begin{proposition}
Let $\phi$ be a value that satisfies Linearity, Positivity, and Null player axiom. Then 

\begin{equation}\label{eq:valform}
\phi_i (v)= \sum_{\substack{S \in \mathcal{D} \\ i \in S^*}} \beta_i(S)\cdot (v(S) - v(S \setminus i)), \qquad  v\in\sfG,\; i\in P,
\end{equation}
for some real numbers $\beta_i(S)\ge 0$. 
\end{proposition}

\begin{proof}
Linearity of $\phi$ and \eqref{game_linear} yield
 $$\phi_i (v) = \sum_{\emptyset \neq S \in \mathcal{D}} \hat{v} (S) \cdot \phi_i (u_S).$$ Since any $i \notin S^*$ is a null player in $u_S$, the sum above becomes $$\phi_i (v) = \sum_{\substack{S \in \mathcal{D} \\ i \in S^*}} \hat{v}(S)\cdot \phi_i (u_S).$$
By Remark \ref{decomposition}, $\phi_i (v)$ is equal to

\[
\sum_{\substack{S \in \mathcal{D} \\ i \in S^*}} \phi_i (u_S) \left( (v(S) - v(S \setminus i)) - \sum_{\substack{T \subset S \\ i \in T^*}} (v(T) - v(T \setminus i)) + \sum_{\substack{T \subset S \\ i \in T^*}} \sum_{\substack{U \subset T \\ i \in U^*}} (v(U) - v(U \setminus i)) - \dots \right)  
\]
Consequently, we can write
\[
\phi_i (v) = \sum_{\substack{S \in \mathcal{D} \\ i \in S^*}} \beta_i(S)\cdot  (v(S) - v(S \setminus i))
\]
for some real numbers $\beta_i(S)$. By way of contradiction, suppose that there exist $i\in P$ and $T\in\calD$ such that $\beta_i(T) < 0$. Then we can find a game $w$ such that $w(S) - w(S \setminus i) = 0$ for every $S \neq T$ and $w(T) - w(T \setminus i) > 0$. However, $\phi_i (w) < 0$, which is a contradiction. 
\end{proof}

\subsection{Hierarchical value} \label{sec:hie}
The hierarchical value was proposed by Faigle and Kern in \cite{FaigleKern92} for coalitional games in which the player set is any partially ordered set and the feasible coalitions are down-sets of the player set as in our setting. The hierarchical value is based on counting the rankings in which a player scores highest among the players in a given coalition. Specifically, for each $S\in\calD$ and every $i\in S$, define
$$
e_S(i)=|\{f\in\calL(P)\mid f(i)>f(j) \text{ for all $j\in S$}\}|
$$
and
\begin{equation}\label{def:hs}
 h_S(i)=\frac{e_S(i)}{|\calL(P)|}.
\end{equation}
The ratio $h_S(i)$ is called the \emph{hierarchical strength} of player $i$ in $S$. Observe that $h_S(i)> 0$ if, and only if, the player $i$ is a maximal element of $S$. Moreover,
$$
\sum_{i\in S} h_S(i) = \frac{1}{|\calL(P)|}\sum_{i\in S} e_S(i) =\frac{1}{|\calL(P)|}\cdot |\calL(P)|= 1,
$$
since every $f\in\calL(P)$ is maximized over $S$ at some element $i\in S^*$, and any two linear extensions attaining their maxima at distinct elements of $S^*$ are necessarily different. We define the hierarchical value as the sharing value \eqref{def:sharingv} where the corresponding sharing system $q$ is 
$$
q(S,i)= \begin{cases}
    h_S(i) & i\in S^*, \\
    0 & i\notin S^*, 
\end{cases}
\qquad S\in\calD,\; i\in P.
$$

\begin{definition}[\cite{FaigleKern92}]\label{def:HV}
The \emph{hierarchical value} is the sharing value $\psi\colon\sfG\to \R^P$ defined by
$$
\psi_i(v)=\sum_{\substack{S\in\calD\\ i \in S^*}}  h_S(i)\cdot  \hat{v}(S), \qquad i\in P, \; v\in\sfG.
$$
\end{definition}
\noindent
For example, the hierarchical value of players in a unanimity game $u_T$ is
\begin{equation}\label{hie_unanimity}
    \psi_i(u_T)=\begin{cases}
        h_S(i) & i\in T^*,\\
        0 & i\in P\setminus T^*.
    \end{cases}
\end{equation}
By \cite[Theorem 1]{FaigleKern92}, the hierarchical value $\psi$ is the only value satisfying Efficiency, Linearity, Null player axiom, and Hierarchical strength axiom, which is defined as follows.
\begin{description}
    \item[Hierarchical Strength] Let $\phi\colon \sfG\to\R^P$ be a value. For any nonempty $S\in\calD$ and players $i,j\in S$, $$h_S(i)\cdot \phi_j(u_S)=h_S(j)\cdot \phi_i(u_S).$$ 
\end{description}
It can be shown that the uniqueness of hierachical value $\psi$ no more holds if Hierarchical Strength is replaced by the (weaker) Symmetry axiom.

It might not be immediately clear from Definition \ref{def:HV} that $\psi$ is a positive value. Positivity follows directly from the representation of $\psi$ as the random-order value in sense of \eqref{eq:RO}:

\begin{equation}\label{HV:ROV}
    \psi_i(v) =\frac{1}{|\calL(P)|} \sum_{f\in\calL(P)}  \Delta_i^f(v), \qquad i\in P, \; v\in\sfG.
    \end{equation}

Now, consider any set $S\subseteq P$ with the partial order $\preceq$ of $P$ restricted to~$S$. 
We define $e_S=|\{f \mid \text{$f$ is a linear extension of $S$}\}|$. Note that $e_P=|\calL(P)|$.  The hierachical value can also be expressed as the average of marginal contributions of player $i$ to every feasible coalition $S$ in which $i$ is maximal:
\begin{equation}\label{eq:PV}
    \psi_i(v) = \sum_{\substack{S\in\calD\\ i \in S^*}} \frac{e_{S\setminus i}\cdot e_{P\setminus S}}{e_P}\cdot (v(S)-v(S\setminus i)).
\end{equation}
\begin{remark}\label{rem:PVs}
    The formula \eqref{eq:PV} shows that the hierarchical value $\psi$ can be viewed as the so-called \emph{probabilistic value}. This family of values was extensively studied in the standard model of coalitional games \cite{Weber88}. In that context, the class of probabilistic values was characterized as the family of values satisfying Linearity, Positivity, and Dummy player axiom \cite[Theorem 5]{Weber88}. However, a probabilistic value lacks Efficiency. The class of efficient probabilistic values, the so-called \emph{quasivalues}, coincides with random-order values in the standard model; see \cite[Theorem 5]{Weber88}.
\end{remark}

\subsection{Other value concepts} \label{sec:other}
There are many different approaches to the definition of coalitional games, feasible coalitions, and values for the games in which players form a hierarchy or precedence structure. We mention here some of them briefly, without claiming completeness. For further details, see the survey \cite{Algaba19} or the discussion in \cite{Beal21}.

The \emph{hierarchical solution} is introduced in \cite{Algaba17}. The hierachical solution is defined as a certain average of Harsanyi dividends. However, this value concept is not a sharing value in the sense of \eqref{def:sharingv}, since the averaging goes over all the coalititons to which the player belongs and not only over the coalitions where the player is maximal. The main differences of hierarchical solution and the hierarchical value of Faigle and Kern \cite{FaigleKern92} are pointed out in \cite{Algaba17}.

Another class of values is constructed for the games in which every coalition of players  is feasible, so that the coalitional game is defined for all subsets $S\subseteq P$.  This makes the distinction  between such games and the setting considered in this paper. The case in point are coalitional games with permission structure and \emph{permission value} studied in \cite{Brink96}.  It was pointed out that the permission value is fundamentally different from the hierarchical value, both numerically and conceptually.

In a related stream of research, Béal et al. \cite{Beal21} recently introduced the \emph{priority value} for coalitional games with the priority structure, which is given by any partial order on the player set $P$. The priority value can be axiomatized. We can formally introduce the priority value in our setting as the sharing value where each Harsanyi divident is split uniformly among the maximal players. However, the resulting sharing value is not positive --- see Example \ref{ex:nonpos}.

\section{Conclusions}\label{sec:concl}
The present paper initiates the study of values in the special game-theoretic setting motivated by the information decomposition problem. We provided a common framework for different solution concepts in case the player set forms a boolean algebra. In particular, we focus on the class of sharing values, which are efficient, but not necessarily positive. Proposition 2 characterizes positive sharing values. Random-order values are described as average selector values by the admissible distributions over the selectors (Proposition 3). We also identify the necessary form for any value satisfying Linearity, Positivity, and Null player axiom (Proposition 5).

We will briefly mention several items for further research. Instead of studying the positivity of sharing values, which are efficient, it is possible to introduce probabilistic values, which are necessarily positive, and try to characterize when the latter are efficient; see also Remark \ref{rem:PVs}.
Note that many results in this paper do not depend on the assumption that the player set is a boolean algebra, and any partial order on the player set can be considered instead. We leave this more general framework for future investigation. The complexity results about enumerating linear extensions indicate that an efficient algorithm to compute the hierachical value or other random-order values cannot be ever found; see also the remark in \cite[p. 260]{FaigleKern92}. Therefore it seems inevitable to focus on the numerical methods to approximate the values using sampling techniques, similar to the existing methods for the classical Shapley value.

\section*{Acknowledgements}
This work has been supported from the GA\v{C}R grant project GA21-17211S and from the project RCI (CZ.02.1.01/0.0/0.0/16\_019/0000765).

\appendix
\section{Coalitional games for the $4$-player boolean algebra}\label{a:4}
Let the player set be $P=\{\bot,a,b,\top\}$. The boolean algebra $(P,\preceq)$ has atoms $a$ and $b$ so its rank is $2$. There are only $2$ linear extensions of $(P,\preceq)$. The lattice of feasible coalitions $\calD$ is on the right-hand side of Figure \ref{FigBA4}. We recall that the notation $\langle a,b \rangle$ denotes the down-set $\{a,b,\bot\}$ whose maximal elements are $a$ and $b$. Consider a coalitional game $v$ on $\calD$ with $v(\langle \bot \rangle)=0$. We write  briefly $v\braket{a,b}$ in place of $v(\braket{a,b})$ to denote the values of $v$. The corresponding Harsanyi dividends are in Table \ref{tab:H4}.
\begin{figure}[h]
    \begin{center}
\begin{tikzcd}[tips=false,column sep=1em,row sep=1.5em]
    & \top \ar{dl} \ar{dr} & \\
    a\ar{dr} & & b\ar{dl} \\
    & \bot \\
\end{tikzcd}
\hspace*{1cm}
\begin{tikzcd}[tips=false,column sep=1em,row sep=1.5em]
   & {P} \ar{d} & \\
   & {\langle a,b\rangle} \ar{dl} \ar{dr}  \\
   {\langle a \rangle} \ar{dr} && {\langle b\rangle} \ar{dl} \\
   & {\langle \bot\rangle} \ar{d} \\
   & {\emptyset} 
\end{tikzcd}
\caption{The $4$-player boolean algebra $(P,\preceq)$ and the lattice of its down-sets $\calD(P,\preceq)$}
\label{FigBA4}
\end{center}
\end{figure}
\begin{table}[h]
    \centering
    \begin{tabular}{|c|c|} \hline
        \emph{Coalition} $S$  & $\hat{v}(S)$ \\ \hline
        $\langle \bot\rangle$ & $0$ \\ 
        $\langle a \rangle$ & $v\braket{a}$ \\
        $\langle b \rangle$ & $v\braket{b}$ \\
        $\langle a,b \rangle$ & $v\braket{a,b}-v\braket{a}-v\braket{b}$ \\
        $P$ & $v(P)-v\braket{a,b}$ \\\hline
    \end{tabular}
    \caption{Harsanyi dividends for a $4$-player game $v$}
    \label{tab:H4}
\end{table}

\section{Coalitional games for the $8$-player boolean algebra}\label{a:8}
The $8$-player boolean algebra $(P,\preceq)$ and its lattice of feasible coalitions $\calD$ are depicted in Figure \ref{FigBA8}. There are $48$ linear extensions of $(P,\preceq)$. We consider a game $v$ over $\calD$ with $v(\braket{\bot})=0$. Table \ref{tab:H8} shows the Harsanyi dividends where we omit the commas in expressions such as $\langle a,b,c \rangle$ for brevity.
\begin{figure}[ht]
    \begin{center}
\begin{tikzcd}[tips=false,column sep=1em,row sep=1.5em]
    & \top \ar{dl} \ar{d} \ar{dr} & \\
    d \ar{d} \ar{dr} & e \ar{dl} \ar{dr} & f \ar{dl} \ar{d} \\
    a\ar{dr} & b\ar{d} & c\ar{dl} \\
    & \bot \\
\end{tikzcd}
\hspace*{1cm}
\begin{tikzcd}[tips=false,column sep=3em,row sep=1.5em]
    & { P } \ar{d} & \\
    & {\langle d,e,f\rangle } \ar{dl} \ar{d} \ar{dr} & \\
    {\langle d,e\rangle } \ar{d} \ar{dr} & {\langle d,f\rangle } \ar{dl} \ar{dr} & {\langle e,f\rangle } \ar{dl} \ar{d} \\
    {\langle c,d\rangle }\ar{dr}\ar{dd} & {\langle b,e\rangle }\ar{d}\ar[bend left=70]{dd} & {\langle a,f\rangle }\ar{dl}\ar{dd} \\
    & {\langle a,b,c\rangle }\ar[bend right=37]{dd}   \ar{ddl} \ar{ddr} & \\
    {\langle d\rangle } \ar{d} & {\langle e\rangle } \ar{d} & {\langle f\rangle } \ar{d} \\
    {\langle a,b\rangle } \ar{d} \ar{dr} & {\langle a,c\rangle } \ar{dl} \ar{dr} & {\langle b,c\rangle } \ar{dl} \ar{d} \\
    {\langle a\rangle }\ar{dr} & {\langle b\rangle }\ar{d} & {\langle c\rangle }\ar{dl} \\
    & {\langle \bot\rangle} \ar{d} \\
    & {\emptyset} 
\end{tikzcd}
\caption{The $8$-player boolean algebra $(P,\preceq)$ and the lattice of its down-sets $\calD(P,\preceq)$}
\label{FigBA8}
\end{center}
\end{figure}
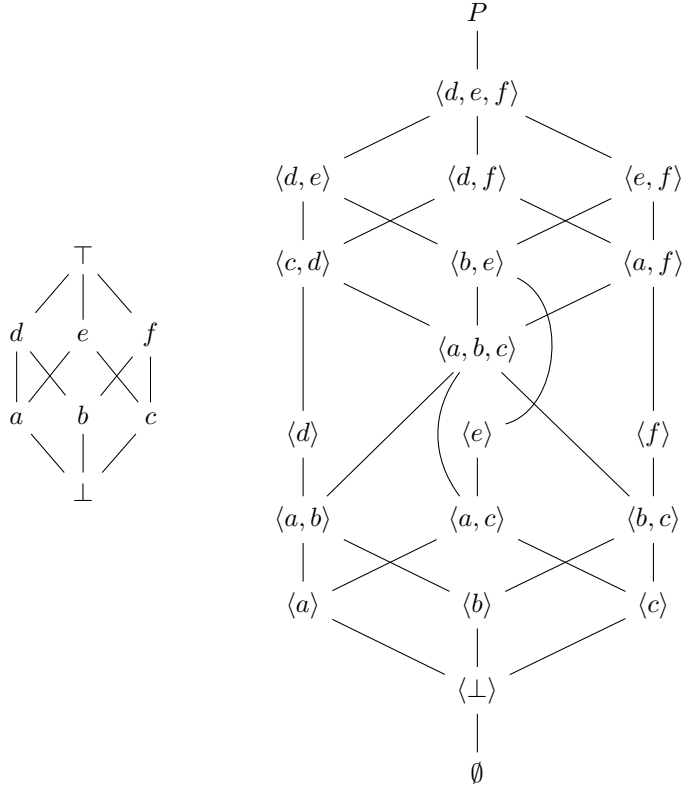

\begin{table}[ht]
    \centering
    \begin{tabular}{|c|c|} \hline
        \emph{Coalition} $S$ & $\hat{v}(S)$ \\ \hline
        $\braket{\bot}$ & $0$ \\ \hline
        $\langle a \rangle$ & $v\braket{a}$ \\
        $\langle b \rangle$ & $v\braket{b}$ \\
        $\langle c \rangle$ & $v\braket{c}$ \\ \hline
        $\langle ab \rangle$ & $v\braket{ab}-v\braket{a}-v\braket{b}$ \\
        $\langle ac \rangle$ & $v\braket{ac}-v\braket{a}-v\braket{c}$ \\
        $\langle bc \rangle$ & $v\braket{bc}-v\braket{b}-v\braket{c}$ \\\hline
        $\langle d \rangle$ & $v\braket{d}-v\braket{ab}$ \\
        $\langle e \rangle$ & $v\braket{e}-v\braket{ac}$ \\
        $\langle f \rangle$ & $v\braket{f}-v\braket{bc}$ \\
        $\langle abc \rangle$ & $v\braket{abc}-v\braket{ab}-v\braket{ac}-v\braket{bc}+v\braket{a}+v\braket{b}+v\braket{c}$ \\ \hline
        $\langle cd \rangle$ & $v\braket{cd}-v\braket{abc}-v\braket{d}+v\braket{ab}$ \\
        $\langle be \rangle$ & $v\braket{be}-v\braket{abc}-v\braket{e}+v\braket{ac}$ \\
        $\langle af \rangle$ & $v\braket{af}-v\braket{abc}-v\braket{f}+v\braket{bc}$ \\ \hline
        $\langle de \rangle$ & $v\braket{de}-v\braket{cd}-v\braket{be}+v\braket{abc}$ \\
        $\langle df \rangle$ & $v\braket{df}-v\braket{cd}-v\braket{af}+v\braket{abc}$ \\
        $\langle ef \rangle$ & $v\braket{ef}-v\braket{be}-v\braket{af}+v\braket{abc}$ \\ \hline
        $\langle def \rangle$ & $v\braket{def}-v\braket{de}-v\braket{df}-v\braket{ef}+v\braket{cd}+v\braket{be}+v\braket{af}-v\braket{abc}$ \\ \hline
        $\langle h \rangle$ & $v(P)-v\braket{def}$ \\ \hline
    \end{tabular}
    \caption{Harsanyi dividends for an $8$-player game $v$}
    \label{tab:H8}
\end{table}

\newpage
\section{Information Attribution} \label{a:IA}
Information Attribution \cite{AyPolaniVirgo20} is reproduced here for the reader's convenience. It works under the assumptions from Section \ref{sec:assum}. Namely we consider a discrete random vector $\bX=(X_1,\dots,X_n)$ of input random varaibles $X_i$, a discrete target random variable $Y$, and a known joint probability distribution $p$ of $(\bX,Y)$, the so-called true distribution. Let $V=\{1,\dots,n\}$ be the index set for input random variables. For a nonempty set of indices $A\subseteq V$, we will use the short notation
$$
\calX_A=\bigtimes_{i\in A} \calX_i.
$$
Analogously, an element of $\calX_A$ is denoted by $\bx_A=(x_i)_{i\in A}$ and the same notation is used for random vectors, $\bX_A=(X_i)_{i\in A}$. We also use short notations $\calX=\calX_V$, 
$
\bX=\bX_V,
$
and $\bx=(x_1,\dots,x_n)\in\calX$.

By $\Prob$ we denote be the set of all joint probability distributions $q$ of random vector $(\bX,Y)$, that is, the elements of $\Prob$ are the functions $q\colon \calX \times \calY \to [0,1]$ satisfying
$$
  \sum_{(\bx,y)\in\calX\times \calY} q(\bx,y)= 1.
$$
Let $\calP_V$ be the set of all subsets of $V$. Any element $A\in\calP_V$ is termed \emph{predictor} as it is associated with the set of input random variables $\{X_i\mid i\in A\}$. We will proceed with the definition of a coalitional game $v_p$ in the sense of Section \ref{sec:coal} induced by the true distribution $p$. The player set in this game is the powerset of all predictors, $\calP_V$, and the set of all feasible coalitions is the family of corresponding down-sets, $\calD=\calD(\calP_V,\subseteq)$. A coalition in this game is thus a down-set of predictors $\calS\in\calD$. 

The coalitional game $v_p$ depends on the optimal solutions to certain convex optimization problems. Given \emph{any} nonempty set of predictors $\calS\subseteq \calP_V$, consider the convex polytope $\Prob_{\calS}\subseteq \Prob$ of probability distributions $q$ whose marginals $q_{\bX}$ and $q_{\bX_AY}$ coincide with the corresponding marginals of the true distribution $p$, for all $A\in \calS$. Specifically, $\Prob_{\calS}$ is the set of all probability distributions $q\in\Prob$ such that
$$
q_{\bX}(\bx )=p_{\bX}(\bx )=\sum_{y\in\calY}p(\bx,y)
$$ for all $\bx \in\calX $ and
$$
q_{\bX_AY}(\bx_A,y)=p_{\bX_AY}(\bx_A,y)=\sum_{\bx_{\bar{A}}\in\calX_{\bar{A}}}p(\bx_A,\bx_{\bar{A}},y)
$$
 for all $A\in \calS$ and all $(\bx_A,y)\in\calX_A\times \calY$, where $\bar{A}=N\setminus A$. 
Note that if $\calS\subseteq \calT\subseteq \calP_V$, then $\Prob_{\calT} \subseteq \Prob_{\calS}$. Indeed, if $q\in\Prob_{\calT}$, then necessarily  $q_{\bX_AY}=p_{\bX_AY}$ for all $A\in \calS$.

Recall that, for any $\calS\subseteq \calP_V$, the set of maximal elements in $\calS$,
$$
\calS^* = \{A\in\calS \mid \text{ there is no $B\in\calS$ such that $A\subset B$}\},
$$
is an antichain. We claim that
\begin{equation}\label{eq:deltas}
    \Prob_{\calS}=\Prob_{\calS^*}.
\end{equation}
It follows immediately that $\Prob_{\calS}\subseteq \Prob_{\calS^*}$ by the inclusion $\calS^*\subseteq \calS$. Suppose that $q\in\Prob_{\calS^*}$ and let $B\in\calS$. Then there is necessarily some $A\in\calS^*$ such that $A\supseteq B$. Since $q_{\bX_A Y}=p_{\bX_A Y}$ by the hypothesis, the inclusion $A\supseteq B$ immediately implies that $q_{\bX_B Y}=p_{\bX_B Y}$. This proves $q\in \Prob_{\calS}$, so \eqref{eq:deltas} is true. 

 We recall that the \emph{(Shannon) entropy} is  
 $$
 H(q)=-\!\!\!\!\sum_{(\bx ,y)\in\calX \times \calY} q(\bx ,y)\log_2 q(\bx ,y), \qquad q\in\Prob,
 $$
 and the \emph{Kullback-Leibler divergence} (or the \emph{relative entropy}) is
 $$
 D(q_1 \parallel q_2) = \sum_{(\bx,y)\in\calX\times \calY} q_1(\bx,y) \log_2\frac{q_1(\bx,y)}{q_2(\bx,y)}, \qquad q_1,q_2\in\Prob.
 $$
The entropy functional $H\colon \Prob\to [0,\infty)$ is continuous and strictly concave. Therefore, for any down-set $\calS\in\calD$, the maximizer of $H$ over $\Prob_{\calS}$ exists and it is determined uniquely. The resulting probability distribution
 \begin{equation} \label{def:split}
 p^{\calS} = \argmax_{q\in\Prob_{\calS}}  H(q)
 \end{equation}
 is called the \emph{split distribution}. It follows from \eqref{eq:deltas} that
 \begin{equation}\label{splitequal}
     p^{\calS}=p^{\calA}
 \end{equation}
 for all $\calS,\calA\subseteq \calP_V$ such that $\calS^*=\calA$. The equality \eqref{splitequal} explains the use of down-sets $\calS$ (or antichains, equivalently) instead of arbitrary sets of predictors as coalitions.
 \begin{example}[The case of $n=3$ input variables]\label{ex:3}
    We assume $V=\{1,2,3\}$ and consider a random vector $(X_1,X_2,X_3,Y)$ whose true probability distribution is $p$. The set of predictors and the associated lattice of their down-sets are depicted in Figure \ref{FigP3}. 
    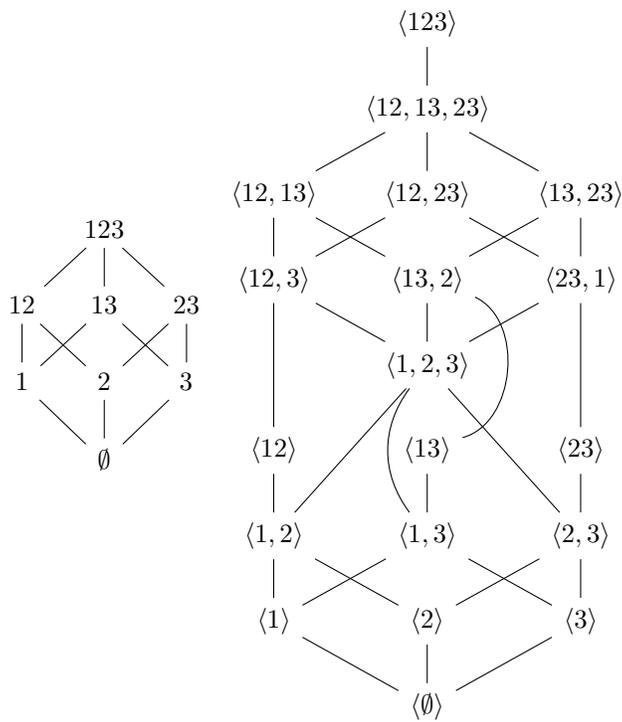
\begin{figure}
        \begin{center}
    \begin{tikzcd}[tips=false,column sep=1em,row sep=1.5em]
        & 123 \ar{dl} \ar{d} \ar{dr} & \\
        12 \ar{d} \ar{dr} & 13 \ar{dl} \ar{dr} & 23 \ar{dl} \ar{d} \\
        1\ar{dr} & 2\ar{d} & 3\ar{dl} \\
        & \emptyset \\
    \end{tikzcd}
    \begin{tikzcd}[tips=false,column sep=1em,row sep=1.5em]
        & {\langle 123\rangle } \ar{d} & \\
        & {\langle 12,13,23\rangle } \ar{dl} \ar{d} \ar{dr} & \\
        {\langle 12,13\rangle } \ar{d} \ar{dr} & {\langle 12,23\rangle } \ar{dl} \ar{dr} & {\langle 13,23\rangle } \ar{dl} \ar{d} \\
        {\langle 12,3\rangle }\ar{dr}\ar{dd} & {\langle 13,2\rangle }\ar{d}\ar[bend left=70]{dd} & {\langle 23,1\rangle }\ar{dl}\ar{dd} \\
        & {\langle 1,2,3\rangle }\ar[bend right=37]{dd}   \ar{ddl} \ar{ddr} & \\
        {\langle 12\rangle } \ar{d} & {\langle 13\rangle } \ar{d} & {\langle 23\rangle } \ar{d} \\
        {\langle 1,2\rangle } \ar{d} \ar{dr} & {\langle 1,3\rangle } \ar{dl} \ar{dr} & {\langle 2,3\rangle } \ar{dl} \ar{d} \\
        {\langle 1\rangle }\ar{dr} & {\langle 2\rangle }\ar{d} & {\langle 3\rangle }\ar{dl} \\
        & {\langle \emptyset\rangle }
    \end{tikzcd}
    \caption{The powerset of predictors and the lattice of down-sets for $3$ input variables.}
    \label{FigP3}
    \end{center}
    \end{figure}
    Let $\calS$ be the down-set $\langle 13,23 \rangle$. Then computing the split distribution \eqref{def:split} over $\Prob_{\langle 13,23 \rangle}$ amounts to solving the following linearly constrained convex optimization problem:
    $$
    \text{Minimize $\sum_{(x_1,x_2,x_3 ,y)\in\calX \times \calY} q(x_1,x_2,x_3 ,y)\log_2 q(x_1,x_2,x_3 ,y)$}
    $$
    subject to the constraints
    \begin{align*}
        q(x_1,x_2,x_3,y)&\geq 0 \qquad \forall x_1,x_2,x_3,y\\
        \sum_{x_1,x_2,x_3,y}  q(x_1,x_2,x_3,y) &= 1\\
        \sum_{y} q(x_1,x_2,x_3,y) &= \sum_{y} p(x_1,x_2,x_3,y) \qquad \forall x_1,x_2,x_3 \\
        \sum_{x_2} q(x_1,x_2,x_3,y) &= \sum_{x_2} p(x_1,x_2,x_3,y) \qquad \forall x_1,x_3,y \\
        \sum_{x_1} q(x_1,x_2,x_3,y) &= \sum_{x_1} p(x_1,x_2,x_3,y) \qquad \forall x_2,x_3,y
    \end{align*}

    \end{example}

    For the problems of smaller size such as the one in Example \ref{ex:3}, the computation of split distribution can be formulated as the optimization on the exponential cone in MOSEK solver. We discuss special cases in which the split distributions has closed-form expression.  In that follows, we frequently omit the curly brackets and the commas in order to use a more compact notation for predictors. For example, we can write $23$ in place of $\{2,3\}$. In the similar spirit $\langle 23,1 \rangle$ denotes the antichain $\{\{2,3\},\{1\}\}$.

\begin{enumerate}
    \item Let $\calS=\calP_V$. We get $\Prob_{\calP_V}=\{p\}$, so the split distribution is  the true distribution $p$ in this case.
    \item Let $\calS=\langle 1,\dots,n\rangle$. Then $\Prob_{\calS}$ contains precisely those $q$ such that $q_{\bX}=p_{\bX}$ and $q_{X_i Y}=p_{X_i Y}$ for all $i\in V$, and the split distribution is the product of one-dimensional marginals of $p$, $$p^{\calS}=p_{X_1}\dotsb p_{X_n}  p_Y.$$
    \item Let $\calS=\{\emptyset\}$. Then $\Prob_{\calS}$ contains precisely those $q$ such that $q_{\bX}=p_{\bX}$ and $q_Y=p_Y$, and the split distribution is given by the product distribution 
    $$
    p^{\calS} = p_{\bX } p_Y.
    $$
\end{enumerate}

 \begin{remark}
    As explained in \cite{AyPolaniVirgo20}, the split distribution $p^{\calS}$ can be also characterized as the unique minimizer of Kullback-Leibler divergence $D(.\parallel u)$ over $\Prob_{\calS}$ from the uniform distribution $u\in \Prob$. Another equivalent formulation leading to the split distribution is based on the minimazation of Kullback-Leibler divergence of the true distribution $p$ from the exponential family; see \cite{Ay18,Rauh}. 
 \end{remark}

 Finally, we are in position to define the coalitional game $v_p$ used in Information Attribution. For any true probability distribution $p\in\Prob$, the coalitional game $v_p$ maps any nonempty down-set $\calS\in\calD$ of predictors to a nonnegative real number 
     \begin{equation}\label{def:game}
         v_p(\calS) = D(p^{\calS} \parallel p_{\bX}p_Y)=\sum_{(\bx,y)\in\calX\times \calY} p^{\calS}(\bx,y) \log_2\frac{p^{\calS}(\bx,y)}{p_{\bX}(\bx)p_Y(y)},
     \end{equation}
     where $p^{\calS}$ is the split distribution and the term on the right-hand-side is the Kullback-Leibler divergence of $p^{\calS}$ from the product distribution $p_{\bX}p_Y$. We put $v_p(\emptyset)=0$.

     Since $p_{\bX}p_Y=p^{\{\emptyset\}}$, the probability distribution $p_{\bX}p_Y$ is precisely the split distribution corresponding to the empty predictor $\emptyset$. Thus, the number $v_p(\calS)$ can be interpreted as the amount of information contained in $p^{\calS}$ in addition to the information already represented by $p_{\bX}p_Y$. Observe that 
     $$
     v_p(\{\emptyset\})=D(p_{\bX}p_Y \parallel p_{\bX}p_Y)=0.
     $$
     Moreover, the assessment of predictor $\calS=\calP_V$ is the mutual information between $\bX$ and $Y$,
     $$
     v_p(\calP_V)=D(p\parallel p_{\bX}p_Y)=\sum_{(\bx,y)\in\calX\times \calY} p(\bx,y) \log_2 \frac{p(\bx,y)}{ p_{\bX}(\bx)p_Y(y)}=I(\bX;Y).
     $$
    
     The natural question is whether the nonnegative coalitional game $v_p$ has some additional properties such as those discussed in Section \ref{sec:coalprop}. It is easy to see that $v_p$ is monotone: If $\calS\subseteq \calT$, then $v_p(\calS)\leq v_p(\calT)$, where we may assume that $\calS\neq\emptyset$. Indeed, by nonnegativity of $D$ and the Pythagorean theorem of information geometry \cite[Theorem 2.8(2)]{Ay18},
     $$
     v_p(\calS)=D(p^{\calS} \parallel  p_{\bX}p_Y) \leq D(p^{\calS} \parallel  p_{\bX}p_Y) + D(p^{\calT} \parallel p^{\calS}) =D(p^{\calT} \parallel  p_{\bX}p_Y)=v_p(\calT).
     $$
   It was shown in \cite[Remark 7.5]{AyPolaniVirgo20} that the game $v_p$ is neither supermodular nor submodular.

   Information Attribution distributes the contribution of predictors $\calS\in\calD$ in the game $v_p$ given by \eqref{def:game} according to the hierachical value discussed in Section \ref{sec:hie}. Using the random-order approach, the contribution of predictor $\calS\in\calD$ is 
   \begin{equation}\label{def:IAvalue}
    \psi_{\calS}(v_p) = \frac{1}{|\calL(\calP_V)|} \sum_{f\in\calL(\calP_V)} \Delta_{\calS}^f(v_p) ,
   \end{equation}
   where $\calL(\calP_V)$ is the set of all linear extensions (admissible permutations) of $\calP_V$, and $\Delta_{\calS}^f(v_p)$ is the marginal contribution \eqref{def:MC} of predictor $\calS$ in game $v_p$ with respect to a linear extension $f$. Other alternative formulas to compute $\psi_{\calS}(v_p)$ are reviewed in Section \ref{sec:hie}. The examples of Information Attribution applied to different true distributions are discussed in \cite[Section 6]{AyPolaniVirgo20}.

\end{document}